\documentclass[12pt,a4paper]{article}
\usepackage{fullpage}
\usepackage{amssymb,amsmath,stmaryrd,amsthm}
\usepackage[mathscr]{eucal}
\usepackage{pstricks,pst-node,pst-text,pst-3d}
\usepackage[round]{natbib}

\usepackage{epsf}
\usepackage{pstricks,pst-node,pst-text,pst-3d}

\theoremstyle{definition}

\newtheorem{theorem}{Theorem}
\newtheorem{proposition}{Proposition}

\theoremstyle{remark}
\newtheorem{remark}{Remark}

\def \bR {\mathbb{R}}
\def \bN {\mathbb{N}}
\def \1{{\mathbf{1}}}
\def \cB {\mathcal{B}}
\def \cC {\mathcal{C}}
\def \cD {\mathcal{D}}
\def \cF {\mathcal{F}}
\def \cJ {\mathcal{J}}
\def \cM {\mathcal{M}}
\def \cN {\mathcal{N}}
\def \cO {\mathcal{O}}
\def \cQ {\mathcal{Q}}

\def \cW {\mathcal{W}}
\def \core {\mathrm{core}}
\def \conv {\mathrm{conv}}
\def \cone {\mathrm{cone}}
\def \ext {\mathrm{ext}}

\begin{document}

\title{The core of games on ordered structures and graphs\footnote{This is an
    updated version of the paper that appeared in 4OR, (2009) 7:207-238.}}

\date{}

\author{Michel GRABISCH\thanks{Paris School of Economics, University of Paris I.
 106-112, Bd. de l'H\^opital, 75013 Paris, France. Email:
\texttt{michel.grabisch@univ-paris1.fr}}}

\maketitle

\begin{abstract}
In cooperative games, the core is the most popular solution concept, and its
properties are well known. In the classical setting of cooperative games, it is
generally assumed that all coalitions can form, i.e., they are all feasible. In many
situations, this assumption is too strong and one has to deal with some unfeasible
coalitions. Defining a game on a subcollection of the power set of the set of
players has many implications on the mathematical structure of the core,
depending on the precise structure of the subcollection of feasible
coalitions. Many authors have contributed to this topic, and we give a unified
view of these different results.

\end{abstract}


\section{Introduction}\label{sec:intro}
Let $N:=\{1,\ldots,n\}$ be a finite set of players. We consider the situation
where these players can form coalitions, and the profit given by the cooperation
of the players in a coalition can be freely distributed among its members: this
is in general referred to as cooperative profit games with transferable utility,
which we will abbreviate in the sequel as \emph{TU-games} (see, e.g.,
\citet{dri88,pesu03,brditi05}).

Let $v$ be a TU-game, that is, a set function $v:2^N\rightarrow\bR$ such that
$v(\emptyset)=0$, assigning to each coalition $S\subseteq N$ its
worth (profit) $v(S)$. Let us assume that forming the grand coalition $N$ is the
best way to generate profit. One of the main problems in
cooperative game theory is to define a rational sharing among all players of
the total worth $v(N)$ of the game. Any sharing is called a \emph{solution} of
the game, since it solves the (difficult) problem of sharing the cake.

The literature on solutions of TU-games is very abundant, and many concepts of
solution have been proposed. One may distinguish among them two main families,
namely those solutions which are single-valued, and those which are
set-valued. In the first category, to each game is assigned a single solution,
which most of the time exists. Best known examples are the Shapley value
\citep{sha53}, the Banzhaf value \citep{ban65} and any other power index used in
voting theory (see, e.g., \citet{fema98}), the nucleolus \citep{sch69}, the
$\tau$-value \citep{tij81}, etc. In the second category, to each game a set of
solutions is assigned. This is the case of the core \citep{gil53,sha71}, the
bargaining set \citep{auma64,dama63}, the kernel \citep{dama65}, etc.

Among all these solution concepts, the core remains one of the most
appealing concepts, due to its simple and intuitive definition. Roughly speaking,
it is the set of sharing vectors for which ``nobody can complain'', or more
exactly, which are coalitionally rational. This means that no coalition can be
better off by splitting from the grand coalition $N$, i.e., for every
$S\subseteq N$, the payoff $x(S)$ given to $S$ is at least equal to $v(S)$, the
profit that $S$ can make without cooperating with the other players. The core
may be empty, but when nonempty, it ensures in some sense the stability of the
set of players, hence its interest. 

The core is an important notion in economics. In an exchange economy, the core
is defined as the set of situations where no coalition of agents can improve the
utility of its members by reassigning the initial resources of its own members
among them \citep{desc63}. Besides, there are many examples in economics where a
common good or resource has to be shared among several users (e.g., a river
supplying the water of several towns). The problem of sharing the cost among all
the users in a rational way precisely amounts to find a solution like the core
\citep{amsp02,brla05,brlava07,khm09}.  The core is also well known in decision
theory and in the field of imprecise probabilities (see the monograph of
\citet{wal91}, and also \citet{chja89}): given a capacity, i.e., a monotonic
game $v$ such that $v(N)=1$ \citep{cho53} representing the uncertainty on the set of states
of nature, its core is the set of probability measures compatible with the
available information on uncertainty. Conversely, given a family of probability
measures representing some uncertainty on the set of states of nature, its lower
envelope defines a capacity.  
  
The core has been widely studied, and its properties are well known. In
particular, when nonempty it is a bounded convex polyhedron, and the famous
Bondareva theorem tells us when the core is nonempty \citep{bon63}, while
 \citet{sha71} and later \citet{ich81} found the vertices of the
core for convex games. 

The classical view in cooperative game theory is to consider that every
coalition can form, i.e., a game $v$ is a mapping defined on $2^N$, the set of
all subsets of $N$. A view closer to reality reveals that it is not always
possible to assume that every coalition can form, so that one should distinguish
between \emph{feasible} and \emph{unfeasible} coalitions. For example, some
hierarchy may exist on the set of players, and feasible coalitions are those
which respect this hierarchy, in the sense that subordinates should be present
(games with precedence constraints, \citet{fake92}). Another example is when
coalitions are the connected subgraphs of a communication graph, depicting who
can communicate with whom \citep{mye77a}. More simply, when considering political
parties, leftist and rightist parties cannot in general make alliance. In fact,
many authors have studied the case where the set of feasible coalitions is a
subcollection of $2^N$, as this paper will show.

The study of the core under such a general framework becomes much more
difficult. Surprisingly, even if the core, when nonempty, is still a convex
polyhedron, it need not be bounded, and moreover, it need not have vertices.
The structure of the core for convex games, perfectly clear in the classical
case, is complicated by the fact that it is not always possible to speak of
convex games in the usual sense, because the definition of convexity works for a
collection of feasible coalitions closed under union and intersection. The aim
of this survey, which is an updated version of \citep{gra09a}, is precisely to
give a unified view of the scattered results around these questions.

The paper is organized as follows. Section~\ref{sec:prer} introduces the basic
material on partially ordered sets and polyhedra. Then Section~\ref{sec:setsys}
is devoted to a comparative study of various families of set systems
(collections of feasible coalitions). Section~\ref{sec:core} defines the core
and the positive core, and gives the main classical results that are valid when all
coalitions are feasible. Section~\ref{sec:struccore} studies the structure of the
core under various assumptions on the set system, while Section~\ref{sec:poco} does
the same for the positive core. Finally, Section~\ref{sec:graph} studies the case of
communication graphs. 

Throughout the paper, the following notation will be used: we denote by $\bR_+$
the set of nonnegative real numbers; $N=\{1,\ldots,n\}$ is the set of players;
for any subset $S\subseteq N$, $\1_S$ denotes the characteristic function (or
vector) of $S$. For singletons, pairs, etc., we often omit braces and commas to avoid a
heavy notation: we write $S\setminus i$, $123$ instead of $S\setminus \{i\}$ and
$\{1,2,3\}$.

\section{Some prerequisites}\label{sec:prer}
\subsection{Partially ordered sets}\label{sec:poset}
The reader can consult, e.g., \citet{dapr90}, \citet{bir67}, and \citet{gratz98}
for more details. A \emph{partially ordered set} (or \emph{poset} for short)
$(P,\leq)$ (or simply $P$ if no confusion occurs) is a set $P$ endowed with a
partial order $\leq$ (i.e, a reflexive, antisymmetric and transitive binary
relation). As usual, $x<y$ means $x\leq y$ and $x\neq y$, while $x\geq y$ is
equivalent to $y\leq x$. Two elements $x,y\in P$ are \emph{incomparable}, and we
denote this by $x||y$, if neither $x\leq y$ nor $y\leq x$ hold.  A useful
example of poset in this paper is $(2^N,\subseteq)$. We say that $x$ \emph{is
covered by} $y$, and we write $x\prec y$, if $x<y$ and there is no $z\in P$ such
that $x< z< y$. A \emph{chain} in $P$ is a sequence of elements $x_1,\ldots,
x_q$ such that $x_1<\cdots < x_q$, while in an \emph{antichain}, any two
elements are incomparable. A chain from $x_1$ to $x_q$ is \emph{maximal} if no
other chain can contain it, i.e., it is a sequence of elements $x_1,\ldots, x_q$
such that $x_1\prec\cdots \prec x_q$. The \emph{length} of a chain is its number
of elements minus 1.

A subset $Q\subseteq P$ is a \emph{downset} of $P$ if for any $x\in Q$, $y\leq x$ implies
$y\in Q$ (and similarly for an upset)\footnote{Some authors use instead the
words \emph{ideals} and \emph{filters}. This is however incorrect, since in the
standard terminology, an ideal is a downset closed under supremum, and a filter is
an upset closed under infimum.}. The set of all downsets of $P$ is denoted by
$\cO(P)$. For any $x\in P$, $\downarrow\!x:=\{y\in P\mid y\leq x\}$ is the
downset generated by $x$ (often called \emph{principal ideal}). 

An element $x\in P$ is \emph{maximal} if there is no $y\in P$ such that $y>x$ (and
similarly for a \emph{minimal} element). $x\in P$ is the (unique)
\emph{greatest} element (or \emph{top} element) of $P$ if $x\geq y$ for all
$y\in P$ (and similarly for the \emph{least} element, or \emph{bottom}
element). Suppose $P$ has a least element $\bot$. Then $x$ is an \emph{atom} of $P$
if $x\succ \bot$. Let $Q\subseteq P$. The element $x\in P$ is an \emph{upper bound} of $Q$
if $x\geq y$ for all $y\in Q$ (and similarly for a \emph{lower bound}). For
$x,y\in P$, the \emph{supremum} of $x,y$, denoted by $x\vee y$, is the least
upper bound of $\{x,y\}$, if it exists (and similarly for the \emph{infimum} of
$x,y$, denoted by $x\wedge y$).

A poset $L$ is a \emph{lattice} if every pair of elements
$x,y\in L$ has a supremum and an infimum. A lattice $L$ is \emph{distributive} if
$\vee,\wedge$ obey distributivity, that is, $x\vee(y\wedge z)=(x\vee
y)\wedge(x\vee z)$ or equivalently $x\wedge(y\vee z)=(x\wedge y)\vee(x\wedge
z)$, for all $x,y,z\in L$. If $L$ is finite, then it has a least and a greatest
element, which we denote by $\bot,\top$ respectively. An element $x\neq \bot$ is
\emph{join-irreducible} if it cannot be expressed as a supremum of other
elements, or equivalently, if it covers only one element. Atoms are
join-irreducible elements. We denote by $\cJ(L)$ the set of all join-irreducible
elements. An important result which will be useful in the sequel is the theorem
of \citet{bir33}: it says that if the lattice $(L,\leq)$ is distributive,
then it is isomorphic to $\cO(\cJ(L))$, where it is understood that $\cJ(L)$ is
endowed with $\leq$, and that the set of downsets is endowed with
inclusion. Conversely, any poset $P$ generates a distributive lattice
$\cO(P)$. This is illustrated on Figure~\ref{fig:birk}
\begin{figure}[htb]
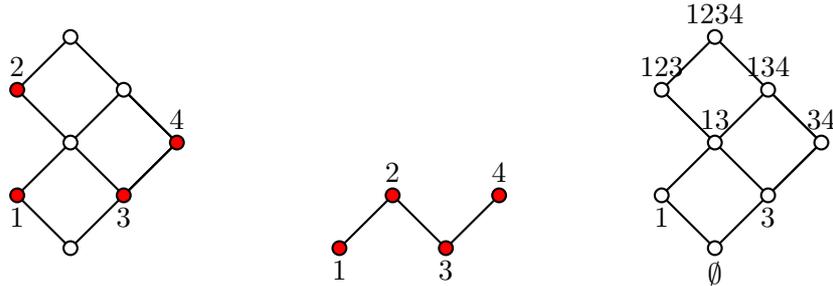

\begin{center}
\psset{unit=0.7cm}
\pspicture(0,0)(3,4)
\pspolygon(1,0)(3,2)(2,3)(0,1)
\pspolygon(2,1)(3,2)(1,4)(0,3)
\pscircle[fillstyle=solid](1,0){0.15}
\pscircle[fillstyle=solid,fillcolor=red](0,1){0.15}
\pscircle[fillstyle=solid](1,2){0.15}
\pscircle[fillstyle=solid,fillcolor=red](2,1){0.15}
\pscircle[fillstyle=solid,fillcolor=red](3,2){0.15}
\pscircle[fillstyle=solid,fillcolor=red](0,3){0.15}
\pscircle[fillstyle=solid](2,3){0.15}
\pscircle[fillstyle=solid](1,4){0.15}
\uput[-90](0,1){\small $1$}
\uput[-90](2,1){\small $3$}
\uput[90](3,2){\small $4$}
\uput[90](0,3){\small $2$}
\endpspicture
\hspace*{2cm}
\psset{unit=0.7cm}
\pspicture(0,0)(3,4)
\psline(0,0)(1,1)(2,0)(3,1)
\pscircle[fillstyle=solid,fillcolor=red](0,0){0.15}
\pscircle[fillstyle=solid,fillcolor=red](1,1){0.15}
\pscircle[fillstyle=solid,fillcolor=red](2,0){0.15}
\pscircle[fillstyle=solid,fillcolor=red](3,1){0.15}
\uput[-90](0,0){\small $1$}
\uput[90](1,1){\small $2$}
\uput[-90](2,0){\small $3$}
\uput[90](3,1){\small $4$}
\endpspicture
\hspace*{2cm}
\psset{unit=0.7cm}
\pspicture(0,0)(3,4)
\pspolygon(1,0)(3,2)(2,3)(0,1)
\pspolygon(2,1)(3,2)(1,4)(0,3)
\pscircle[fillstyle=solid](1,0){0.15}
\pscircle[fillstyle=solid](0,1){0.15}
\pscircle[fillstyle=solid](1,2){0.15}
\pscircle[fillstyle=solid](2,1){0.15}
\pscircle[fillstyle=solid](3,2){0.15}
\pscircle[fillstyle=solid](0,3){0.15}
\pscircle[fillstyle=solid](2,3){0.15}
\pscircle[fillstyle=solid](1,4){0.15}
\uput[-90](1,0){\small $\emptyset$}
\uput[-90](0,1){\small $1$}
\uput[-90](2,1){\small $3$}
\uput[90](0,3){\small $123$}
\uput[90](3,2){\small $34$}
\uput[90](1,2){\small $13$}
\uput[90](2,3){\small $134$}
\uput[90](1,4){\small $1234$}
\endpspicture
\end{center}
\caption{Left: a distributive lattice $L$. Join-irreducible elements are those in
  grey. Middle: the poset $\cJ(L)$ of join-irreducible elements. Right: the set
  $\cO(\cJ(L))$ of all downsets of $\cJ(L)$ ordered by inclusion, which is
  isomorphic to $L$}
\label{fig:birk}
\end{figure}   

Let $P$ be a poset, and $x\in P$. The \emph{height} of $x$ is the length of a
longest chain in $P$ from a minimal element to $x$.  The \emph{height} of a
lattice $L$ is the height of its top element, i.e., it is the length of a
longest chain from bottom to top. When the lattice is distributive, its height
is $|\cJ(L)|$.

\subsection{Inequalities and polyhedra}\label{sec:poly}
We recall only some basic facts useful for the sequel (see, e.g., \citet{zie95},
\citet{fakest02} for details). Our exposition mainly follows \citet[\S 1.2]{fuj05b}

We consider a set of linear equalities and inequalities with real constants
\begin{align}
\sum_{j=1}^n a_{ij}x_j & \leq b_i \quad (i\in I)\label{eq:p1}\\
\sum_{j=1}^n a'_{ij}x_j & = b'_i \quad (i\in E).\label{eq:p2}
\end{align}
This system defines an intersection of halfspaces and hyperplanes, called a (closed convex)
polyhedron.  A set $C\subseteq \bR^n$ is a \emph{convex cone} (or simply a
\emph{cone}) if $x,y\in C$ implies that $\alpha x +\beta y\in C$ for all
$\alpha,\beta\geq 0$ (\emph{conic} combination). The cone is \emph{pointed} if
$C\cap (-C)=\{0\}$ (equivalently, if it has an extreme point, see below).  An
\emph{affine set} $A$ is the translation of a subspace of the vector space
$\bR^n$. Its dimension is the dimension of the subspace. A \emph{line} is a
one-dimensional affine set, and a \emph{ray} is a ``half-line'', i.e., a
translation of a set given by $\{\alpha x\mid \alpha\geq 0\}$ for some
$x\in\bR^n$, $x\neq 0$. An \emph{extreme ray} of a cone is a ray whose supporting vector
cannot be expressed as a convex combination of the supporting vectors of other
rays. Any cone can be expressed as the conic combination of its extreme rays. An
\emph{extreme point} or \emph{vertex} of a polyhedron $P$ is a point in $P$
which cannot be expressed as a convex combination of other points in $P$. A
polyhedron is pointed if it contains an extreme point. The \emph{recession cone}
$C(P)$ of a polyhedron $P$ defined by (\ref{eq:p1}) and (\ref{eq:p2}) is defined
by
\begin{align}
\sum_{j=1}^n a_{ij}x_j & \leq 0 \quad (i\in I)\label{eq:p3}\\
\sum_{j=1}^n a'_{ij}x_j & = 0 \quad (i\in E).\label{eq:p4}
\end{align}
The recession cone is either a pointed cone (possibly reduced to $\{0\}$) or it
contains a line. The following basic properties are fundamental:
\begin{enumerate}
\item $P$ has rays (but no line) if and only if $C(P)$ is a pointed cone
different from $\{0\}$;
\item  $P$ is pointed  if and only if $C(P)$ does not contain a line, or
  equivalently, if the system (\ref{eq:p4}) and
\[
\sum_{j=1}^n a_{ij}x_j  = 0 \quad (i\in I)
\]
has 0 as unique solution.
\item $P$ is a \emph{polytope} (i.e., a bounded polyhedron)  if and only if $C(P)=\{0\}$.
\end{enumerate}
The fundamental theorem of polyhedra asserts that any pointed polyhedron $P$ defined by
a system (\ref{eq:p1}) and (\ref{eq:p2}) is the Minkowski sum of its recession
cone (generated by its extreme rays; this is the conic part of $P$) and the
convex hull of its extreme points (the convex part of $P$):
\[
P = \cone(r_1,\ldots, r_k) + \conv(\ext(P))
\]
where $r_1,\ldots,r_k$ are the extreme rays of $C(P)$, $\cone()$ and $\conv()$
indicate respectively the set of all conic and convex combinations, and $\ext()$ is the set
of extreme points of some convex set.  

If $P$ is not
pointed, then it reduces to its recession cone up to a translation. 

\medskip

Finally, suppose that in the system (\ref{eq:p1}) and (\ref{eq:p2}) defining a
polyhedron $P$, the equalities in (\ref{eq:p2}) are independent (i.e., $P$ is
$(n-|E|)$-dimensional). A \emph{$p$-dimensional face} ($0\leq p\leq n-|E|$) of
$P$ is a set of points in $P$ satisfying in addition $q=n-|E|-p$ independent
equalities in (\ref{eq:p1}). In particular, $P$ itself is a face of $P$ ($q=0$),
a \emph{facet} is a $(n-|E|-1)$-dimensional face ($q=1$), and a vertex is a
0-dimensional face ($q=n-|E|$). Clearly, no vertex can exist (i.e., $P$ is not
pointed) if $|I|<n-|E|$.

\section{Set systems}\label{sec:setsys}
Our study deals with games defined on a collection of feasible coalitions. In this
section, we introduce various possible structures for these collections. The
weakest requirement we introduce is that the collection should include the grand
coalition, and for mathematical convenience, the empty set. There are however
exceptions to this rule.  

A \emph{set system} $\cF$ on $N$ is a subset of $2^N$ containing $\emptyset$ and
$N$. Endowed with inclusion,  $\cF$ is a poset with top and bottom elements
$N,\emptyset$ respectively. The set of maximal chains from $\emptyset$ to $N$ in
$\cF$ is denoted by $\cC(\cF)$. For any $S\subseteq N$, we put $\cF(S):=\{T\in
\cF\mid T\subseteq S\}$.

A set system $\cF$ is \emph{atomistic} if $\{i\}\in\cF$ for all $i\in N$.

For any collection $\cF\subseteq 2^N$, we introduce 
\[
\widetilde{\cF}:=\{S\in 2^N\mid S=F_1\cup\cdots\cup F_k, \quad
F_1,\ldots,F_k\in\cF \text{ pairwise disjoint}\}
\]
the family \emph{generated} by $\cF$ \citep{fai89}.

\medskip

Let $\cF$ be a set system. A \emph{TU-game} (or simply \emph{game}) on $\cF$ is a mapping
$v:\cF\rightarrow \bR$ such that $v(\emptyset)=0$. The game is \emph{monotonic}
if for $S,T\in\cF$ such that $S\subseteq T$, we have $v(S)\leq v(T)$ (and
therefore $v$ is nonnegative).

When $\cF=2^N$, the notion of convexity and superadditivity are well known. A
game is said to be \emph{convex} if for any $S,T\in 2^N$, we have
\[
v(S\cup T) + v(S\cap T) \geq v(S) + v(T).
\]
A game $v$ is \emph{superadditive} if the above inequality holds for disjoint
subsets, i.e., for all $S,T\in 2^N$ such that $S\cap T=\emptyset$, 
\[
v(S\cup T)\geq v(S) + v(T).
\] 
The above notions generalize as follows. 
Assume $\cF$ is a (set) lattice. A game $v$ on $\cF$ is \emph{convex} if for any
$S,T\in\cF$,
\[
v(S\vee T) + v(S\wedge T) \geq v(S) + v(T). 
\]
Superadditivity amounts to the above inequality restricted to subsets $S,T$ such
that $S\wedge T=\emptyset$. 
Obviously, one could not speak of convex game if the set system is not a
lattice. It is however possible to find alternative definitions for weaker
structures, as will be seen in the sequel (see Section~\ref{sec:pocoaug}, and
supermodular games in Section~\ref{sec:fagr}).

The \emph{M\"obius transform} of $v$ on $\cF$ is a real-valued mapping $m^v$ on
$\cF$ given implicitely by the system of equations
\[
v(S) = \sum_{F\subseteq S, F\in \cF}m^v(F),\quad S\in\cF.
\] 
As it is well known, when $\cF=2^N$, we obtain $m^v(S) = \sum_{F\subseteq
  S}(-1)^{|S\setminus F|}v(F)$. The M\"obius transform is known as the
  \emph{Harsanyi dividends} \citep{har63} in game theory.

Given these general definitions, we turn to the study of the main families of
set systems.

\subsection{Regular set systems}
Let $1\leq k\leq n$. A set system is \emph{$k$-regular} if all maximal chains
from $\emptyset$ to $N$ have the same length $k$ \citep{xigr09}.  A $n$-regular
set system is simply called a \emph{regular set system}
\citep{hogr05a,lagr06b}. Equivalently, $\cF$ is regular if and only if for
$S,T\in\cF$ such that $S\succ T$, we have $|S\setminus T|=1$.

Any regular set system satisfies:
\begin{enumerate}
\item {\bf One-point extension:} if $S\in \cF$, $S\neq N$, then $\exists i\in
  N\setminus S$ such that $S\cup i\in \cF$; 
\item {\bf Accessibility:} if $S\in \cF$, $S\neq \emptyset$, then $\exists i\in
  S$ such that $S\setminus i\in \cF$.
\end{enumerate}
The converse is not true (see Figure~\ref{fig:1}).
\begin{figure}[htb]
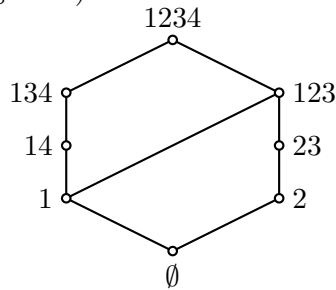

\begin{center}
\psset{unit=0.7cm} 
\pspicture(-0.5,-0.5)(4,4)
\pspolygon(2,0)(0,1)(0,3)(2,4)(4,3)(4,1)
\psline(0,1)(4,3)
\pscircle[fillstyle=solid](2,0){0.1}
\pscircle[fillstyle=solid](0,1){0.1}
\pscircle[fillstyle=solid](0,2){0.1}
\pscircle[fillstyle=solid](0,3){0.1}
\pscircle[fillstyle=solid](2,4){0.1}
\pscircle[fillstyle=solid](4,1){0.1}
\pscircle[fillstyle=solid](4,2){0.1}
\pscircle[fillstyle=solid](4,3){0.1}
\uput[-90](2,0){\small $\emptyset$}
\uput[180](0,1){\small 1}
\uput[180](0,2){\small 14}
\uput[180](0,3){\small 134}
\uput[90](2,4){\small 1234}
\uput[0](4,1){\small 2}
\uput[0](4,2){\small 23}
\uput[0](4,3){\small 123}
\endpspicture
\end{center}
\caption{A set system satisfying one-point extension and accessibility, but
  which is not $k$-regular}
\label{fig:1}
\end{figure}

In a $k$-regular set system $\cF$, for any $S,T\in
\cF$, all maximal chains from $S$ to $T$ have the same
length.
\begin{remark}
Obviously, regular set systems (and to a less extent, $k$-regular set systems)
offer a convenient mathematical framework because all maximal chains have length
$n$, and for this reason many notions (in particular marginal worth vectors, and
therefore the Shapley value \citep{sha53} and the Weber set (see
Section~\ref{sec:core})) can be defined as in the classical case $\cF=2^N$. One
can however find motivations for such structures which are more
game-theoretically oriented:
\begin{enumerate}
\item The set of connected coalitions in a connected communication graph is a
  regular set system (see Section~\ref{sec:graph}). The converse is false: $\{i\}\in\cF$
  for all $i\in N$ is a necessary condition (for necessary and sufficient
  conditions: see augmenting systems in Section~\ref{sec:aug}).
\item Maximal chains correspond to permutations on $N$ (or total orders on
  players). A regular set system forbids some permutations, i.e., some orderings
  of the players to enter the game. With $k$-regular set systems, $k<n$, players
  may enter the game by groups.   
\end{enumerate}
\end{remark}

\subsection{Convex geometries and antimatroids}\label{sec:cogoanti}
A \emph{convex geometry} $\cF$ \citep{edja85} is a collection of subsets of $N$
containing the empty set, closed under intersection, and satisfying the
one-point extension property. Necessarily $N\in \cF$, hence it is a set system,
and moreover a regular set system.

An \emph{antimatroid} $\cF$ \citep{dil40} is a collection of subsets of $N$ containing the
empty set, closed under union, and satisfying the accessibility property. Any
antimatroid satisfies the \emph{augmentation} property:
\[
S,T\in\cF \text{ with } |T|>|S|\Rightarrow \exists i\in T\setminus S \text{
  s.t. } S\cup i\in\cF.
\]
If $\cF$ satisfies $\bigcup \cF=N$, then $N\in\cF$. Such antimatroids are called
\emph{normal} by \citet{bri09}.
\begin{remark}
\citet{albibrji04} relate antimatroids to permission
structures; see Section~\ref{sec:cupcap}. However, the relation is somewhat artificial
since antimatroids do not always correspond to permission structures (this is
the case of systems closed under $\cup,\cap$). The unusual word ``poset
antimatroids'' is used, and means the set of upsets (or downsets) of a
poset. These are antimatroids closed under intersection. But it is well known
that such set systems are distributive lattices $\cO(N)$ (and so could be called poset
convex geometries as well), hence closed under union and
intersection (see Section~\ref{sec:poset}).
\end{remark}

\subsection{Set lattices}\label{sec:lat}
If a set system is a lattice, we call it a \emph{set lattice}. It need not be
closed under $\cap,\cup$, nor be a $k$-regular set system (see for example the
pentagon on Figure~\ref{fig:2a} (ii)).

If the lattice is distributive, then we benefit from Birkhoff's theorem and we
know that it is generated by a poset $P$. However this poset is not always $N$
endowed with some partial order. The following can be easily proved and
clarifies the situation (see \citet{xigr09}):
\begin{proposition}\label{prop:3.3}
Let $\cF$ be a distributive set lattice
on $N$ of height $k$. The following holds.
\begin{enumerate}
\item $\cF$ is a $k$-regular set system, which is generated by a
  poset $P$ of $k$ elements, i.e., $\cF$ is isomorphic to $\cO(P)$.
\item $\cF$ is closed under union and intersection if and only if
  $\cF$ is isomorphic to $\mathcal{O}(P)$, where $P$ can be chosen to be a partition of $N$.
\item $k=n$ if and only if  $\cF$ is isomorphic to $\mathcal{O}(N)$.
\end{enumerate}
\end{proposition}
Figure~\ref{fig:2} shows the relative situation of set lattices and $k$-regular
set systems.
\begin{figure}[htb]
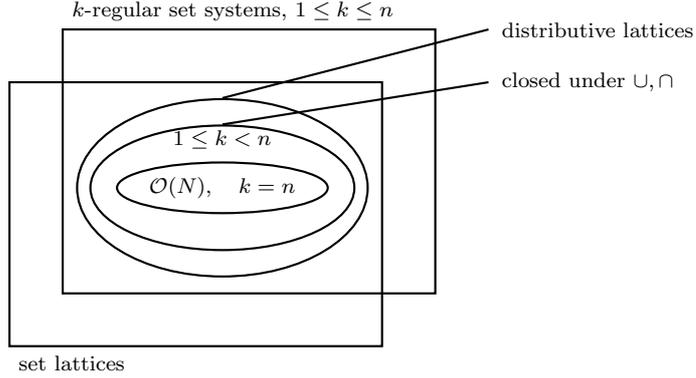

\begin{center}
\psset{unit=0.7cm} 
\pspicture(0,-0.5)(8,7)
\pspolygon(0,0)(0,5)(7,5)(7,0)
\pspolygon(1,1)(1,6)(8,6)(8,1)
\psellipse(4,3)(2.75,1.7)
\psellipse(4,3)(2.5,1.2)
\psellipse(4,3)(2,0.5)
\psline(4,4.7)(9,6)
\psline(4,4.2)(9,5)
\uput[-45](0,0){\scriptsize set lattices}
\uput[45](1,6){\scriptsize $k$-regular set systems, $1\leq k\leq n$}
\uput[0](9,6){\scriptsize distributive lattices}
\uput[0](9,5){\scriptsize closed under $\cup,\cap$}
\uput[90](4,3.5){\scriptsize $1\leq k<n$}
\rput(4,3){\scriptsize $\mathcal{O}(N),\quad k=n$}
\endpspicture
\end{center}
\caption{Set lattices and $k$-regular set systems}
\label{fig:2}
\end{figure}
Figure~\ref{fig:2a} shows that all inclusions are strict.
\begin{figure}[htb]
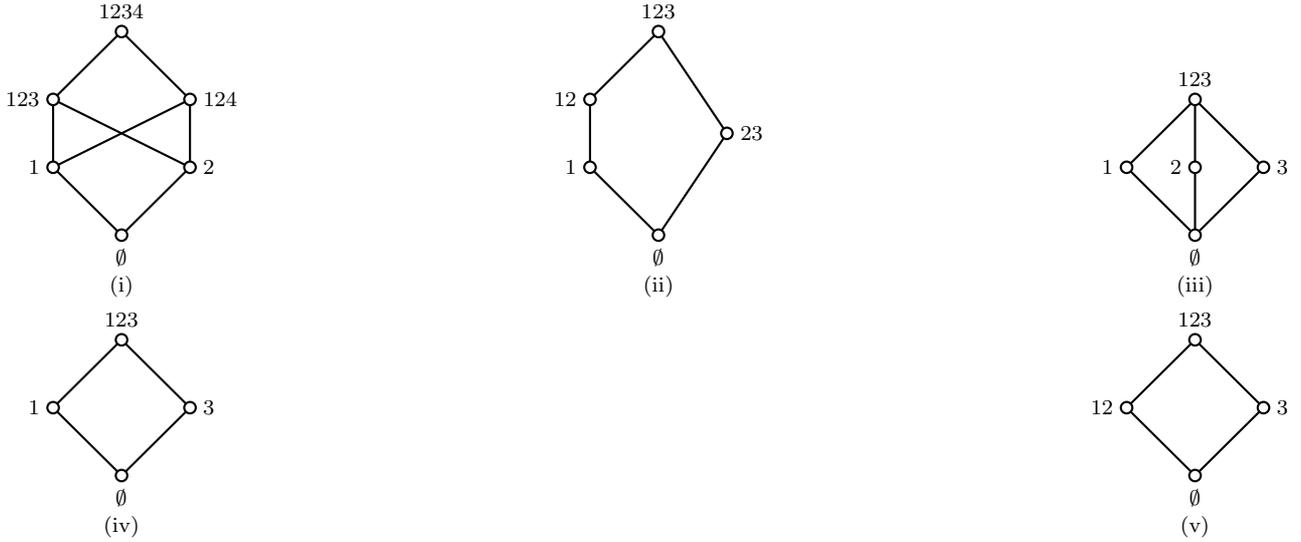

\begin{center}
\psset{unit=0.45cm} \pspicture(0,-2.5)(4,8)
\pspolygon(2,0)(0,2)(0,4)(2,6)(4,4)(4,2) \psline(0,2)(4,4)
\psline(0,4)(4,2) \pscircle[fillstyle=solid](2,0){0.2}
\pscircle[fillstyle=solid](0,2){0.2}
\pscircle[fillstyle=solid](0,4){0.2}
\pscircle[fillstyle=solid](2,6){0.2}
\pscircle[fillstyle=solid](4,4){0.2}
\pscircle[fillstyle=solid](4,2){0.2}
 \uput[-90](2,0){\scriptsize $\emptyset$} 
\uput[-180](0,2){\scriptsize  1} 
\uput[-180](0,4){\scriptsize
123} \uput[90](2,6){\scriptsize 1234}
\uput[0](4,4){\scriptsize 124} \uput[0](4,2){\scriptsize 2}
\rput(2,-1.5){\scriptsize (i)}
\endpspicture
\hfill
\psset{unit=0.45cm} \pspicture(0,-2.5)(4,6)
\pspolygon(2,0)(0,2)(0,4)(2,6)(4,3)
\pscircle[fillstyle=solid](2,0){0.2}
\pscircle[fillstyle=solid](0,2){0.2}
\pscircle[fillstyle=solid](0,4){0.2}
\pscircle[fillstyle=solid](2,6){0.2}
\pscircle[fillstyle=solid](4,3){0.2}
\uput[-90](2,0){\scriptsize $\emptyset$}
\uput[180](0,2){\scriptsize 1}
\uput[180](0,4){\scriptsize 12}
\uput[90](2,6){\scriptsize 123}
\uput[0](4,3){\scriptsize 23}
\rput(2,-1.5){\scriptsize (ii)}
\endpspicture
\hfill
\psset{unit=0.45cm}
\pspicture(0,-2.5)(4,4)
\pspolygon(2,0)(0,2)(2,4)(4,2)
\psline(2,0)(2,4)
\pscircle[fillstyle=solid](2,0){0.2}
\pscircle[fillstyle=solid](0,2){0.2}
\pscircle[fillstyle=solid](2,4){0.2}
\pscircle[fillstyle=solid](4,2){0.2}
\pscircle[fillstyle=solid](2,2){0.2}
\uput[-90](2,0){\scriptsize $\emptyset$}
\uput[-180](0,2){\scriptsize $1$}
\uput[-180](2,2){\scriptsize $2$}
\uput[0](4,2){\scriptsize $3$}
\uput[90](2,4){\scriptsize $123$}
\rput(2,-1.5){\scriptsize (iii)}
\endpspicture

\psset{unit=0.45cm}
\pspicture(0,-2)(4,4.5)
\pspolygon(2,0)(0,2)(2,4)(4,2)
\pscircle[fillstyle=solid](2,0){0.2}
\pscircle[fillstyle=solid](0,2){0.2}
\pscircle[fillstyle=solid](2,4){0.2}
\pscircle[fillstyle=solid](4,2){0.2}
\uput[-90](2,0){\scriptsize $\emptyset$}
\uput[-180](0,2){\scriptsize $1$}
\uput[0](4,2){\scriptsize $3$}
\uput[90](2,4){\scriptsize $123$}
\rput(2,-1.5){\scriptsize (iv)}
\endpspicture
\hfill
\psset{unit=0.45cm}
\pspicture(0,-2)(4,4.5)
\pspolygon(2,0)(0,2)(2,4)(4,2)
\pscircle[fillstyle=solid](2,0){0.2}
\pscircle[fillstyle=solid](0,2){0.2}
\pscircle[fillstyle=solid](2,4){0.2}
\pscircle[fillstyle=solid](4,2){0.2}
\uput[-90](2,0){\scriptsize $\emptyset$}
\uput[-180](0,2){\scriptsize $12$}
\uput[0](4,2){\scriptsize $3$}
\uput[90](2,4){\scriptsize $123$}
\rput(2,-1.5){\scriptsize (v)}
\endpspicture
\end{center}
\caption{(i) $k$-regular but not a lattice; (ii) lattice but not $k$-regular;
  (iii) $k$-regular lattice but not distributive; (iv) distributive lattice but
  not closed under $\cup$; (v) closed under $\cup,\cap$ but not isomorphic to $\cO(N)$}
\label{fig:2a}
\end{figure}

\subsection{Systems closed under union and intersection}\label{sec:cupcap}
As seen in Section~\ref{sec:lat}, these are particular set lattices, which are
distributive and generated by a partition of $N$.

\citet{degi95} proved that they are equivalent to
\emph{(conjunctive) permission structures} of \citet{giowbr92}.
A (conjunctive) permission structure is a mapping $\sigma:N\rightarrow 2^N$ such that
$i\not\in \sigma(i)$. The players in $\sigma(i)$ are the direct subordinates of
$i$. ``Conjunctive'' means that a player $i$ has to get the permission to act of all
his superiors. Consequently, an \emph{autonomous} coalition contains all
superiors of every member of the coalition, i.e., the set of autonomous
coalitions generated by the permission structure $\sigma$ is
\[
\cF_\sigma=\{S\in 2^N\mid S\cap \sigma(N\setminus S)=\emptyset\}.
\]
Then $\cF$ is closed under union and intersection if and only if $\cF=\cF_\sigma$ for some
permission structure $\sigma$.

See also \citet{albibrji04} for similar results related to
antimatroids (see Section~\ref{sec:cogoanti}). They characterize \emph{acyclic}
permission structures (i.e., where, for all $i\in N$, in the set of all
subordinates (not limited to the direct ones) of $i$, $i$ is not present) by
distributive lattices $\mathcal{O}(N)$ (called there poset antimatroids).

\subsection{Weakly union-closed systems}\label{sec:weucl}
A set system $\cF$ is \emph{weakly union-closed} if $F\cup F'\in\cF$ for all $F,F'\in \cF$
  such that $F\cap F'\neq\emptyset$.

An important consequence is that for any $S\subseteq N$, $\cF(S):=\{F\in \cF\mid
F\subseteq S\}$ has pairwise disjoint maximal elements.

The \emph{basis} of $\cF$ is the collection of sets $S$ in $\cF$ which cannot be
written as $S=A\cup B$, with $A,B\in\cF$, $A,B\neq S$, $A\cap B\neq\emptyset$
\citep[Chap. 6]{bil00}. All singletons and pairs of $\cF$ are in the basis.
Clearly, knowing the basis permits to recover $\cF$.

\begin{remark}\label{rem:wuc}
\begin{enumerate}
\item This terminology  is used by  \citet{fagr09}. Weakly union-closed systems have been
  studied under the name \emph{union stable systems} by  \citet{alg98}
  (summarized in \citet[Chap. 6]{bil00}). 
\item They are closely related to communication graphs because if $\cF$ represents
a communication graph (i.e., $\cF$ is the collection of connected coalitions of the graph;
see Section~\ref{sec:graph}),
then the union of two intersecting connected coalitions must be
connected. \citet{bri09} characterized those weakly union-closed collections
which correspond to communication graphs: $\cF\subseteq 2^N$ is the set of
connected coalitions of some comunication graph if and only if
$\emptyset\in\cF$, $\cF$ is normal (i.e., $\bigcup\cF=N$), weakly union-closed,
and satisfies \emph{2-accessibility} (i.e., $S\in\cF$ with $|S|>1$ implies that
there exist distinct $i,j\in S$ such that $S\setminus i$ and $S\setminus j$
belong to $\cF$). Another characterization is due to Bilbao through augmenting
systems (see Section~\ref{sec:aug}).    
\item Similarly, they are also related to the more general notion of
  \emph{conference structures} of \citet{mye80}, which generalize communication
  graphs. A conference structure $\cQ$ is a collection of subsets of $N$ of
  cardinality at least 2. Two players $i,j$ are connected if there is a sequence
  $S_1,\ldots,S_k$ of sets in $\cQ$ such that $i\in S_1,j\in S_k$, and
  $S_\ell\cap S_{\ell+1}\neq\emptyset$ for $\ell=1,\ldots,k-1$. Then,
  $\cF:=\{S\subseteq N\mid \forall i,j\in S,\text{ $i$ and $j$ are connected}\}$
  is a weakly union-closed system. Conversely, given a weakly union-closed
  system $\cF$, the basis of $\cF$ restricted to sets of cardinality at least 2
  can be considered as a conference structure. An equivalent view of this is
  given by van den Nouweland et al. through hypergraphs \citep{noboti92}, since
  in a hypergraph, a (hyper)link joins several nodes (and thus can be viewed as
  a subset of cardinality at least 2). Thus, a path in a hypergraph corresponds
  to a sequence $S_1,\ldots,S_k$ as described above.
\end{enumerate}
\end{remark} 

\subsection{Partition systems}
They were studied by \citet[\S 5.1]{bil00} and \citet{albilo01}.  A \emph{partition system}
is a collection $\cF\subseteq 2^N$ containing the empty set, all singletons, and
such that for every $S\subseteq N$, the maximal subsets of $S$ in $\cF$ form a
partition of $S$ (equivalently, $\cF$ contains $\emptyset$, all singletons and
is weakly union-closed).

Any set system induced by a communication graph is a partition system.
If $\cF$ is a partition system, then $\widetilde{\cF}=2^N$.

\subsection{Augmenting systems}\label{sec:aug} 
An \emph{augmenting system} \citep{bil03,bior08,bior08a} is a
collection $\cF\subseteq 2^N$ containing $\emptyset$, being weakly union-closed,
and satisfying
\[
\forall S,T\in \cF \text{ s.t. } S\subseteq T, \quad \exists i\in T\setminus
S\text{ s.t. } S\cup i\in \cF.
\] 
\begin{remark}\label{rem:augm}
\begin{enumerate}
\item In \citet{bil03}, it is required in addition that $\bigcup \cF=N$ (obviously,
  this property should always be required when dealing with collections of subsets).
\item  $N$ does not necessarily belong to $\cF$. If $N\in\cF$, the above property
implies that all maximal chains from $\emptyset$ to $N$ have the same length
$n$, and thus $\cF$ is a regular set system. The converse is false.

If $N\not\in \cF$, by weak union-closure, all maximal sets in $\cF$, say
$F_1,\ldots,F_k$, are disjoint, and no $F\in \cF$ can intersect two distinct
maximal subsets. Therefore, $\cF$ can be partitioned into augmenting subsystems
$\cF_1,\ldots,\cF_k$ on $F_1,\ldots, F_k$ respectively, which are all
regular. Hence, it is sufficient to study the case where $N\in \cF$. 
\item An augmenting system is an antimatroid (respectively, convex geometry) if
  and only if $\cF$ is closed under union (respectively, intersection).
\item Augmenting systems are of particular importance since they permit to
  characterize communication graphs (see Section~\ref{sec:graph}). Specifically,
  if $G$ is a communication graph, the set of connected coalitions is an
  augmenting system. Conversely, an augmenting system is the collection of
  connected coalitions of a communication graph if $\{i\}\in\cF$ for all $i\in
  N$. Each connected component of the graph corresponds to the augmenting
  subsystems $\cF_1,\ldots,\cF_k$ mentionned in (ii).
\end{enumerate} 
\end{remark}
Augmenting systems are also closely related to the previously introduced
structures, as shown in the next proposition.
\begin{proposition}
$\cF$ is an augmenting system containing $N$ if and only if it is regular and weakly
union-closed.
\end{proposition}
\begin{proof}
The ``only if'' part has been already noticed above.  Now, suppose it is
regular. Take $S\subseteq T$ in $\cF$, then they lie on some chain. Since
the system is regular, there are $t-s-1$ subsets between $S$ and $T$, which
implies the augmentation property.
\end{proof}
The various relations between regular set systems, weakly union-closed systems
and augmenting systems are illustrated on  Figure~\ref{fig:3}. 
\begin{figure}[htb]
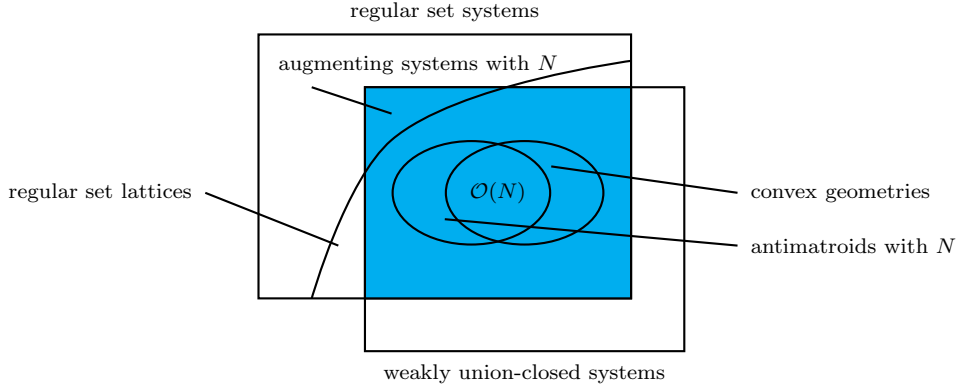

\begin{center}
\psset{unit=0.7cm} 
\pspicture(0,-0.5)(8,7)
\pspolygon(0,1)(0,6)(7,6)(7,1)
\pspolygon(2,0)(2,5)(8,5)(8,0)
\pspolygon[fillstyle=solid,fillcolor=cyan](2,1)(2,5)(7,5)(7,1)
\psellipse(4,3)(1.5,1)
\psellipse(5,3)(1.5,1)
\pscurve(1,1)(2.5,4)(7,5.5)
\uput[90](3.5,6){\scriptsize regular set systems}
\uput[-90](5,0){\scriptsize weakly union-closed systems}
\psline(1,5)(2.5,4.5)
\uput[90](3,5){\scriptsize augmenting systems with $N$}
\psline(1.5,2)(-1,3)
\uput[180](-1,3){\scriptsize regular set lattices}
\psline(3.5,2.5)(9,2)
\uput[0](9,2){\scriptsize antimatroids with $N$}
\psline(5.5,3.5)(9,3)
\uput[0](9,3){\scriptsize convex geometries}
\rput(4.5,3){\scriptsize $\cO(N)$}
\endpspicture
\end{center}
\caption{Regular set systems and weakly union-closed systems}
\label{fig:3}
\end{figure}
Figure~\ref{fig:4} shows that it is possible to have regular set lattices
  which are not weakly union-closed, and weakly union-closed regular systems not
  being a lattice.
\begin{figure}
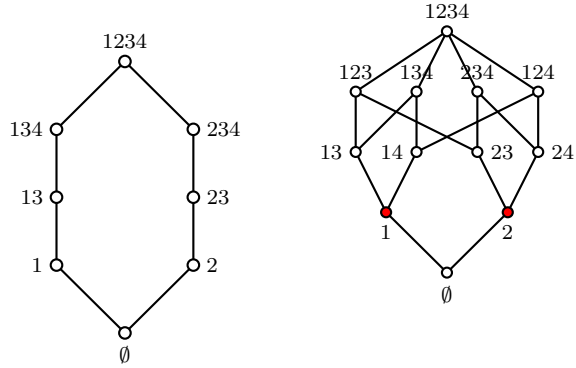

\begin{center}
\psset{unit=0.45cm} \pspicture(0,0)(4,8)
\pspolygon(2,0)(0,2)(0,6)(2,8)(4,6)(4,2)
\pscircle[fillstyle=solid](2,0){0.2}
\pscircle[fillstyle=solid](0,2){0.2}
\pscircle[fillstyle=solid](0,4){0.2}
\pscircle[fillstyle=solid](0,6){0.2}
\pscircle[fillstyle=solid](2,8){0.2}
\pscircle[fillstyle=solid](4,6){0.2}
\pscircle[fillstyle=solid](4,4){0.2}
\pscircle[fillstyle=solid](4,2){0.2}
\uput[-90](2,0){\scriptsize $\emptyset$}
\uput[180](0,2){\scriptsize 1}
\uput[180](0,4){\scriptsize 13}
\uput[180](0,6){\scriptsize 134}
\uput[90](2,8){\scriptsize 1234}
\uput[0](4,6){\scriptsize 234}
\uput[0](4,4){\scriptsize 23}
\uput[0](4,2){\scriptsize 2}
\endpspicture
\hspace*{2cm}
\psset{unit=0.4cm}
\pspicture(-3,-2)(3,10)
\pspolygon(0,0)(-2,2)(-3,4)(-3,6)(0,8)(3,6)(3,4)(2,2)
\psline(-2,2)(-1,4)(3,6)
\psline(2,2)(1,4)(-3,6)
\psline(-3,4)(-1,6)(0,8)
\psline(3,4)(1,6)(0,8)
\psline(-1,4)(-1,6)
\psline(1,4)(1,6)
\pscircle[fillstyle=solid](0,0){0.2}
\pscircle[fillstyle=solid,fillcolor=red](-2,2){0.2}
\pscircle[fillstyle=solid](-3,4){0.2}
\pscircle[fillstyle=solid](-3,6){0.2}
\pscircle[fillstyle=solid](-1,6){0.2}
\pscircle[fillstyle=solid](0,8){0.2}
\pscircle[fillstyle=solid](3,6){0.2}
\pscircle[fillstyle=solid](1,6){0.2}
\pscircle[fillstyle=solid](3,4){0.2}
\pscircle[fillstyle=solid,fillcolor=red](2,2){0.2}
\pscircle[fillstyle=solid](-1,4){0.2}
\pscircle[fillstyle=solid](1,4){0.2}
\uput[-90](0,0){\scriptsize $\emptyset$}
\uput[-90](-2,2){\scriptsize $1$}
\uput[-90](2,2){\scriptsize $2$}
\uput[-180](-3,4){\scriptsize $13$}
\uput[-180](-1,4){\scriptsize $14$}
\uput[0](1,4){\scriptsize $23$}
\uput[0](3,4){\scriptsize $24$}
\uput[90](-3,6){\scriptsize $123$}
\uput[90](-1,6){\scriptsize $134$}
\uput[90](3,6){\scriptsize $124$}
\uput[90](1,6){\scriptsize $234$}
\uput[90](0,8){\scriptsize $1234$}
\endpspicture
\end{center}
\caption{Left: regular set lattice but not weakly union-closed. Right: regular
  and weakly union-closed but not a lattice, since 1 and 2 have no supremum}
\label{fig:4}
\end{figure}

\subsection{Coalition structures}
Our last category is of different nature since it is not a set system in our
sense, and its motivation is very different from the notion of feasible
coalition. Its origin comes from the domain of coalition formation. We
nevertheless mention
it here due to its importance, although the topic of coalition formation
falls outside the scope of this survey (see, e.g., \citet{hoow01}, \citet{dee97}). 

A \emph{coalition structure} on $N$ is a partition of $N$
\citep{audr74}. It is called by \citet{owe77} \emph{a priori union structures}.

\medskip

Let $v$ be a game on $2^N$, and consider a coalition structure
$\cB:=\{B_1,\ldots,B_m\}$.

Given a payoff vector $x$, $B_k\in\cB$, we define the game $v^*_x$ on $2^{B_k}$ by
\[
v_x^*(S) = \begin{cases}
  \max_{T\subseteq N\setminus B_k}(v(S\cup T) - x(T)), & S\subset B_k,
  S\neq\emptyset\\
v(S), & S=B_k\text{ or } \emptyset.
  \end{cases}
\]

\begin{remark}
The above definition considers a game in the classical sense, and not on the
blocks of the partition. Another type of game suited for coalition structures is
global games and games in partition function form. We mention them without
further development since their nature is too far away from coalitional games. A
\emph{global game} \citep{gile91} is a real-valued mapping defined on the set
$\Pi(N)$ of partitions of $N$: it assigns some worth to any coalition
structure. A \emph{game in partition function form} \citep{thlu63} is a mapping
$v$ assigning a real number to a coalition $S$ and a partition $\pi$ containing
$S$: knowing that the coalition structure is some $\pi\in\Pi(N)$ such that
$S\in\pi$, $v(S,\pi)$ is the worth of $S$ in this organization of the society.
\end{remark}

\section{The core and related notions}\label{sec:core}
Let $v$ be a game on a set system $\cF$. A \emph{payoff vector} is any $x\in
\bR^n$. It represents some amount of money given to the players. By commodity we
write $x(S):=\sum_{i\in S}x_i$ for any $S\subseteq N$. A payoff vector is
\emph{efficient} if $x(N)=v(N)$. The \emph{pre-imputation set} of $v$ is the set
of all efficient payoff vectors. We define the \emph{imputation set} of $v$ as
\[
I(v):=\{x\in\bR^n\mid x_i\geq v(\{i\}) \text{ if } \{i\}\in\cF \text{ and } x(N)=v(N)\}.
\]  
The \emph{core} of $v$ is defined by
\[
\core(v):=\{x\in\mathbb{R}^n \mid x(S)\geq v(S)\text{ for all }S\in
\cF\text{ and }x(N)=v(N)\}.
\]
The \emph{positive core} of $v$ is defined by
\citep{fai89}:
\[
\core^+(v):=\{x\in\bR_+^n\mid x(S)\geq v(S) \text{ for all } S\in\cF\text{ and }
x(N)=v(N)\}. 
\]
\begin{remark}
\begin{enumerate}
\item The classical definition of the core \citep{sha71} is recovered with
$\cF=2^N$. It should be noted that the definition is meaningful only if the game
is a profit game. For cost games, the inequalities should be reversed. The core
is the set of payoff vectors which are coalitionally rational: no coalition can
have a better profit if it splits from the grand coalition $N$. 
\item If $\cF$ contains all singletons and the game is monotonic, the
  distinction between the core and the positive core is void since they
  coincide. The positive core retains imputations which are nonnegative, which
  means that in no case the players would have to pay something instead of being
  rewarded. This notion is natural essentially if the game is monotonic, for
  otherwise there would exist players with negative contribution to the game,
  and those players should be penalized; see also Remark~\ref{rem:6.1}.
\item The (positive) core is also well known in decision under risk and uncertainty: it is
  the set of probability measures dominating a given capacity (monotonic game);
  see, e.g., \cite{wal91}. More surprisingly, the  (positive) core with the reversed
  inequalities is a well-known concept in combinatorial optimization, under the
  name of \emph{base polyhedron of a polymatroid} \citep{edm70}, where a
  polymatroid is nothing else than a submodular\footnote{When $\cF$ is closed
  under $\cup,\cap$, a \emph{submodular} game, also
  called \emph{concave}, satisfies the inequality $v(S\cup T) + v(S\cap T) \leq
  v(S) + v(T)$ for all $S,T\in 2^N$.} monotonic game. As we will see, many
  theorems shown in game theory about the core were already known in combinatorial
  optimization; see the excellent monograph of \citet{fuj05b}.
\end{enumerate}
\end{remark}
We recall the classical results on the core when
$\cF=2^N$. 

A first important question is to know whether the core is empty or not. A
collection $\cB$ of nonempty subsets of $2^N$ is
\emph{balanced} if there exist positive coefficients $\lambda_B, B\in\cB$,
such that
\[
\sum_{B\in\cB} \lambda_B\1_{B}=\1_N.
\]
The vector $\lambda:=(\lambda_B)_{B\in\cB}$ is called the \emph{balancing vector}.
Balanced collections generalize the notion of partitions. \citet{depe98} have
shown that a collection is balanced if and only if for all $y\in\bR^n$ such that
$y(N)=0$ (side-payment vector), if $y(S)>0$ for some $S\in\cB$, then it exists
$S'\in\cB$ such that $y(S')<0$. 

A game
$v$ on $2^N$ is \emph{balanced} if for every balanced collection $\cB$ with
balancing vector $\lambda$ it holds
\[
\sum_{B\in \cB}\lambda_Bv(B)\leq v(N).
\]
This could be interpreted by saying that there is no advantage in dividing the
grand coalition into balanced collections. The following well-known result is an
easy consequence of the duality theorem of linear programming (or Farkas Lemma).
\begin{theorem}
\citep{bon63} Let $v$ be a game on $2^N$. Then $\core(v)\neq\emptyset$ if and
only if $v$ is balanced.
\end{theorem}
 
A strong version exists where only minimal balnced collections are used.
Obviously, the interest of the theorem is more mathematical than
algorithmical. It cannot reasonably be used for testing the nonemptiness of the
core of a given game.  Since the core is a set of linear inequalities, classical
tools testing the feasability of a set of inequalities, like the Fourier-Motzkin
elimination, or simply a linear programming solver, should be used. 
 
Assuming that the core is nonempty, it is a polytope, and therefore the question
of knowing its vertices arises.  To each maximal chain $C\in\mathcal{C}(2^N)$ with
$C=\{\emptyset,S_1,\ldots,S_n=N\}$, corresponds bijectively a permutation
$\sigma\in\mathfrak{S}(N)$, the set of permutations  on $N$, such that 
\[
S_i=\{\sigma(1),\ldots,\sigma(i)\}, \quad i=1,\ldots,n.
\]
Considering a game $v$ on $2^N$, to each permutation $\sigma$ (or maximal
  chain~$C$) we assign a \emph{marginal worth vector}
  $\phi^\sigma$ (or $\phi^{C}$) in $\mathbb{R}^n$ by:
\[
\phi^\sigma_{\sigma(i)} :=v(S_i)-v(S_{i-1}) = v(S_{i-1}\cup {\sigma(i)}) - v(S_{i-1}).
\]
The \emph{Weber set} is the convex hull of all marginal worth vectors:
\[
\cW(v):=\conv(\phi^C \mid C\in\cC(2^N)).
\]
The following inclusion always holds
\[
\core(v)\subseteq \cW(v).
\]
\begin{theorem}
The following assertions are equivalent.
\begin{enumerate}
\item $v$ is convex
\item All marginal vectors $\phi^C$, $C\in \cC(2^N)$ (or $\phi^\sigma$, $\sigma\in \mathfrak{S}(N)$), belong
to the core
\item $\core(v) = \mathrm{conv}(\{\phi^\sigma\}_{\sigma\in\mathfrak{S}(N)})$
\item $\mathrm{ext}(\core(v)) = \{\phi^\sigma\}_{\sigma\in\mathfrak{S}(N)}$.
\end{enumerate}
\end{theorem}
\citet{sha71} proved (i)$\Rightarrow$(ii) and (i)$\Rightarrow$(iv), while
\citet{ich81} proved (ii)$\Rightarrow$(i). \citet{edm70} proved the same result
as Shapley. This is also mentionned in  \citep{lov83}.  This result
clearly shows why convexity is an important property for games. Indeed, in this
case, the core is nonempty and its structure is completely known. In the
subsequent sections, we will see that much effort is done for games defined on
set systems in order to preserve these properties as far as possible.

\section{Structure of the core}\label{sec:struccore}
\subsection{General results for arbitrary set systems}\label{sec:gene}
We begin by some simple considerations on the imputation set. If $\cF$ is
atomistic, then $I(v)\neq\emptyset$ if and only if $v(N)\geq \sum_{i\in
N}v(\{i\})$. If $\cF$ is not atomistic, then it is always true that $I(v)\neq
\emptyset$. Indeed, if $\{j\}\not\in\cF$, just take $x_i=v(i)$ if $\{i\}\in\cF$,
$x_j=v(N)-\sum_{\{i\}\in\cF}v(i)$, and $x_i=0$ otherwise.  Similarly, $I(v)$ is
bounded if and only if $\cF$ is atomistic.

\medskip

The first question we address is to know when the core is nonempty. It is easy
to see that the classical definitions and result of Bondareva on balancedness
still work: $\core(v)\neq\emptyset$ if and only if $v$ is balanced, where
balanced collections are understood as collections in $\cF$. Another result is
due to \citet{fai89}, with a different (but equivalent) definition of balancedness. 
A game $v$ on $\cF$ is \emph{balanced} if for all families $A_1,\ldots,A_k$ in
$\cF$ and $m\in\bN$ it holds
\[
\frac{1}{m}\sum_{i=1}^k \1_{A_i}=\1_N\text{ implies }
\frac{1}{m}\sum_{i=1}^k v(A_i)\leq v(N).
\]
In the above, it should be noted that repetitions are allowed in the family and
that the length of a family is arbitrary.

Assuming that $\core(v)$ is nonempty, one can define its \emph{lower envelope}
$v_*$, which is a game on $2^N$:
\[
v_*(S) := \min_{x\in\core(v)}x(S), \quad\forall S\subseteq N.
\]
Note that $v_*(N)=v(N)$, and if $\cF=2^N$, we have $\core(v_*)=\core(v)$. 
\begin{remark}
The lower envelope is an important notion in decision theory (see
\citet{wal91}). In game theory, it is called the \emph{Harsanyi mingame} \citep{delava08}.
\end{remark}
An important question is to know whether the equality $v=v_*$ holds. Such games
are called \emph{exact}. If a game $v$ is exact, by the above mentioned
property, it is the smallest game having the core equal to $\core(v)$.  
\citet{fai89} proved the next result:
\begin{theorem}
A game $v$ is exact if and only if for all families
$A,A_1,\ldots,A_k$ in $\cF\setminus \{\emptyset\}$ and $m,l\in\bN$,
\[
\sum_{i=1}^k\1_{A_i} = m\1_N + l\1_A\text{ implies } \sum_{i=1}^kv(A_i)\leq
mv(N) + lv(A).
\]
\end{theorem}
As above, repetitions are allowed in the family.
This is similar to a result of \citet{sch72}, proved when $\cF$ is a (possibly
infinite) family closed under union and complementation: $v$
is exact if and only if for all $S\in \cF$
\begin{multline*}
v(S) = \sup\Big\{\sum_i a_iv(S_i) - a|v| \text{ s.t. } \sum_i a_i\1_{S_i} - a\1_N\leq
\1_S,\\ \text{ with } a\in \bR_+, (a_i,S_i) \text{ is a finite sequence in }
\bR_+\times \cF\Big\},
\end{multline*}
and $|v|:=\sup\{\sum_ia_iv(S_i)\mid (a_i,S_i) \text{ is a finite sequence in }
\bR_+\times \cF\text{ s.t.} \sum_ia_i\1_{S_i}\leq \1_N\}$.

\medskip

When nonempty, the core is a polyhedron. Therefore it makes sense to speak of its
recession cone (proposed under the name of \emph{core of the set system $\cF$}
by \citet{dere98}, hence the notation):
\[
\core(\cF) :=\{x\in\mathbb{R}^n
\mid x(S) \geq 0 \text{ for all }S\in
\cF\text{ and }x(N)=0\}.
\]
A direct application of results of Section~\ref{sec:poly} leads to:
\begin{enumerate}
\item $\core(v)$ has rays if and only if $\core(\cF)$ is a pointed cone
different from $\{0\}$.  Then $\core(\cF)$ corresponds to the conic part of
$\core(v)$;
\item $\core(v)$ has no vertices if and only if $\core(\cF)$ contains a line;
\item  $\core(v)$ is a polytope  if and only if $\core(\cF)=\{0\}$.
\end{enumerate}
Therefore, it remains to study the structure of the recession cone.  We introduce
\[
\mathrm{span}(\cF)
:=\big\{S\subseteq N\mid \1_S=\sum_{T\in \cF}\lambda_T
\1_T\text{ for some }\lambda_T\in\mathbb{R}\big\}.
\]
$\cF$ is \emph{non-degenerate} if $\mathrm{span}(\cF) = 2^N$
\footnote{In fact, it is simpler to check the following equivalent condition:
for all $i\in N$, it exists a linear combination of the $\1_T$'s, $T\in \cF$,
giving $\1_i$.}. Non-degeneracy implies the discerning property
(see Section~\ref{sec:un}). The converse holds if $\cF$ is closed under
$\cup,\cap$ (see Theorem~\ref{th:degi}). We give two easy sufficient conditions
for $\cF$ to be non-degenerate:
\begin{enumerate}
\item $\cF$ contains all singletons (obvious from Footnote~3);
\item $\cF$ is regular. Indeed, since any chain has length $n$, all $\1_i$'s can
  be recovered from $\1_{S_j}-\1_{S_{j-1}}$, for two consecutive sets $S_j,S_{j-1}$ in a
  chain. 
\end{enumerate} 
\begin{theorem}\label{th:nond}
\citep{dere98} $\core(\cF)$ is a pointed cone if and only if $\cF$ is non-degenerate.
\end{theorem}
This result is easy to see from Footnote~2 and
Section~\ref{sec:poly}. Indeed, non-degeneracy is equivalent to the existence of
linear combinations of the $\1_T$'s, $T\in\cF$, giving all $\1_i$'s, $i\in N$,
and the same linear combinations can  be used to express all $x_i$'s, $i\in N$,
from the system $x(T)=0$, $T\in\cF$, thus proving that this system has a unique
solution (which is 0). But this is equivalent to say that the recession cone is a
pointed cone.  

We recall that $\cF$ is \emph{balanced} if $\exists\lambda_S>0$ for all $S\in\cF$ such that
$\1_N=\sum_{S\in\cF}\lambda_S\1_S$. 
\begin{theorem}
\citep{dere98} $\core(\cF)$ is a linear subspace if and only if $\cF$ is balanced.
\end{theorem}
Therefore, $\core(\cF)=\{0\}$ if and only if $\cF$ is balanced and non-degenerate.

\medskip

Lastly, considering a game $v$ on $\cF$, we
introduce the following extension of $v$ to $\widetilde{F}$ \citep{fai89}
\[
\tilde{v}(S):=\max\Big\{\sum_{i\in I}v(F_i), \quad \{F_i\}_{i\in I} \text{ is a
  $\cF$-partition of $S$}\Big\},
\]
where by a $\cF$-partition we mean a partition whose blocks belong to $\cF$.
The game $\tilde{v}$ is superadditive, and if $\tilde{v}(N)=v(N)$, then
$\core(v)=\core(\tilde{v})$, which is easy to show.
\begin{remark}
If $\cF$ contains all singletons (e.g., a partition system), then
$\widetilde{\cF}=2^N$, and so $\tilde{v}$ is an extension of $v$ on $2^N$: compare
with the extension $\overline{v}$ defined in Section~\ref{sec:ext}. Also,
$\tilde{v}$ is a partitioning game of Kaneko and Wooders (see
Section~\ref{sec:part}). If $v$ is a superadditive game on a partition system,
then $\tilde{v}=\overline{v}$. In \citet[\S 5.3]{bil00} it is shown that if
$\cF$ is a partition system containing $N$ and
$v(N)=\overline{v}(N)=\tilde{v}(N)$, then $\core(\overline{v}) =
\core(\tilde{v})$.
\end{remark}

\subsection{Set systems closed under $\cup,\cap$}\label{sec:un}
Let $\cF$ be a set system closed under $\cup,\cap$ (such systems are
distributive lattices, and correspond to permission
structures; see Section~\ref{sec:cupcap}). For each $i\in N$ we define
\[
D_i:=\bigcap\{S\in\cF\mid S\ni i\} = \text{smallest $S$ in $\cF$ containing $i$}.
\]
\begin{proposition}\label{prop:di}
The set of $D_i$'s coincides with the set of join-irreducible elements of $\cF$,
i.e.,
\[
\{D_i\}_{i\in N} = \cJ(\cF).
\]
Moreover, if the height of
$\cF$ is strictly smaller than $n$, necessarily we have $D_i=D_j$ for some $i,j$
(the height equals the number of distinct $D_i$'s). 
\end{proposition}
\begin{proof}
Suppose there is some $D_i$ which is not a join-irreducible element. Then $D_i$
can be written as the supremum of other elements, which are smaller. Since $\cF$
is closed under $\cup$, one of these elements must contain $i$, a contradiction
with the definition of $D_i$.

Conversely, take a join-irreducible element $S$. If $S=\{i\}$, we are
done. Assume then that $|S|>1$. Since it covers only one subset, say $S'$, for any
$i$ in $S\setminus S'$, $S$ is the smallest subset containing $i$, whence the result. 

Finally, since  $\cF$ is a
distributive lattice, its height is the number of join-irreducible elements,
hence some $D_i$'s must coincide if the height is less than $n$.
\end{proof}
\begin{remark}
The sets $D_i$'s are introduced in \citet{degi95}. They are also known in the
  literature of combinatorial optimization (see \citet[Sec. 3.3]{fuj05b} and
  \citet[Sec. 7.2 (b.1)]{fuj05b} (principal partitions)).
\end{remark} 
\begin{theorem}
\[
\core(\cF) = \mathrm{cone}(\1_j-\1_i\mid i\in N \text{ and } j\in D_i).
\]
\end{theorem}
If $\cF$ is not closed under $\cup,\cap$, then any $\1_j-\1_i$ is a ray of $\core(\cF)$.
\begin{remark}
This result is due to \cite{degi95}. It was in fact proved when the system is of
the type $\cO(N)$ in a more precise form by Tomizawa (\citet{tom83}, cited in
\citet[Th. 3.26]{fuj05b}): it says that the extreme rays are those corresponding
to $j\succ i$ in $(N,\leq)$. Note that it could be easily adapted if the lattice
is not generated by $N$, but by a partition of $N$ (see
Proposition~\ref{prop:3.3}). 

Another consequence of this result is that when $\cF$ is closed under
$\cup,\cap$, the core is \textit{always} unbounded, unless $(N,\leq)$ (or the
partition of $N$ endowed with $\leq$) is an antichain.
\end{remark}
The set system $\cF$ is \emph{discerning} if all $D_i$'s are different (equivalently, by
Proposition~\ref{prop:di}, if the height of $\cF$ is $n$, which is a much
simpler condition).
\begin{theorem}\label{th:degi}
\citep{degi95} Consider $ \cF$ to be closed under $\cup,\cap$.
$\core(\cF)$ is a pointed cone if and only if $\cF$ is discerning.
\end{theorem}
This result is easy to deduce from previous facts. If the recession cone is
pointed, then $\cF$ is non-degenerate by Theorem~\ref{th:nond}, which implies the
discerning property as mentionned above. If $\cF$ is discerning, then its height
is $n$, and so it is regular, which implies that it is non-degenerate, and
therefore, the recession cone is pointed.

When $\cF$ is of the type $\cO(N)$, for any maximal chain $C\in\cC(\cF)$, define
the \emph{marginal vector} associated to $C$ like in the classical case, and
define the \emph{Weber set} as the convex hull of all marginal vectors.
\begin{theorem}\label{th:degial}
Let $\cF$ be of the type $\cO(N)$. Then the
convex part of the core is included in the Weber set.
\end{theorem}
\begin{theorem}\label{th:futo}
Let $\cF$ be of the type $\cO(N)$. Then $v$ is convex if and only if the
convex part of the core is equal to the Weber set.
\end{theorem}
\begin{remark}
The two last theorems are shown by \citet{grxi07}, but they can
deduced from \cite{degi95}, where they are stated for acyclic permission
structures. Indeed, from \citet{albibrji04}, we know that these systems are
equivalent to distributive lattices of the type $\cO(N)$ (see
Section~\ref{sec:cupcap}).  The ``only if'' part of the latter theorem was
already shown by \citet{futo83}.
\end{remark}

\medskip

Lastly, we address a slightly more general case, where closure under union is
replaced by weak union-closure. The following development is due to \citet{fai89}.
$A,B\subseteq N$ is a \emph{crossing pair} if $A,B$ intersect, $A\cup B\neq N$
and $A\setminus B,B\setminus A$ are nonempty. Then $\cF$ is a \emph{crossing
  family} if $A\cup B,A\cap B\in\cF$ whenever $A,B$ is a crossing pair. $v$ on a
crossing family $\cF$ is
\emph{convex} if for every crossing pair $v(A\cup B) + v(A\cap B)\geq v(A) +
v(B)$.
\begin{theorem}
Suppose $\cF$ is weakly union-closed and closed under intersection. Then
$\tilde{\cF}$ is closed under union and intersection, and $v$ convex on $\cF$
implies $\tilde{v}$ convex on $\tilde{\cF}$. 
\end{theorem}  
\begin{theorem}
Suppose $\cF$ is weakly union-closed and closed under intersection, and $v$ on
$\cF$ is convex. Then $v$ is balanced if and only if for all partitions (with nonempty
blocks, as usual) $\{A_1,\ldots,A_k\}$
of $N$,
\[
v(A_1) + \ldots v(A_k) \leq v(N).
\]
Now, if the $A_i$'s are only pairwise disjoint, this characterizes complete
balancedness (see Section~\ref{sec:poco}). 
\end{theorem}

\subsection{Distributive lattices generated by a poset on $N$}
Consider $\cF=\cO(N)$ for some partial order $\leq$ on $N$. Then $\cF$ is a
regular distributive set lattice, closed under union and intersection, and
previous results give us the properties of the core. Considering a balanced game
$v$, we have seen that:
\begin{enumerate}
\item $\core(v)$ is a pointed polyhedron, since $\cF$ is non-degenerate (see
  Section~\ref{sec:gene});
\item $\core(v)$ is unbounded, unless $(N,\leq)$ is an antichain.
\item $\core^F(v)\subseteq \cW(v)$, where $\core^F(v)$ is the convex part of
$\core(v)$ (see Theorem~\ref{th:degial});
\item $v$ is convex if and only if $\core^F(v)=\cW(v)$ (see Theorem~\ref{th:futo}).
\end{enumerate}
The following result \citep{grsu12} shows that all games can be balanced, provided $(N,\leq)$
is connected (i.e., all players are within a single hierarchy).
\begin{proposition}
If $(N,\leq)$ is connected, then $\core(v)\neq\emptyset$ for any game $v$ on
$\cO(N)$. Conversely, if $(N,\leq)$ is not connected, there exists some game $v$
on $\cO(N)$ such that $\core(v)=\emptyset$. 
\end{proposition}

\medskip

So far we have mainly studied extremal points and rays.  Although faces of the
core have not drawn the attention of game theorists, they have been deeply
studied in combinatorial optimization when $\cF=\cO(N)$ (see \citet[Ch. 2, \S
  3.3 (d)]{fuj05b} for a detailed account). We restrict here to basic facts.

Assuming $v$ is a balanced game, take any $x\in\core(v)$ and define $\cF(x)
=\{S\in\cF\mid x(S) = v(S)\}$. Then $\cF(x)$ is a sublattice of $\cF$ if $v$ is
convex. Indeed, first remark that $\emptyset,N\in\cF(x)$. Now, take $S,T\in
\cF(x)$ and let us prove that $S\cup T$ and $S\cap T$ belong to
$\cF(x)$. Assuming it is false, we have
\begin{align*}
x(S) + x(T)  & = x(S\cup T) + x(S\cap T)   > \\ & v(S\cup T) + v(S\cap T) \geq v(S) +
v(T)  = x(S) + x(T),
\end{align*}
a contradiction.

Assuming $v$ is convex, define for any subsystem $\cD\subseteq \cF$
\begin{align*}
F(\cD) & := \{x\mid x(S) = v(S), \forall S\in \cD, \quad x(S)\geq v(S) \text{
  otherwise}\}\\
F^\circ(\cD) & := \{x\mid x(S) = v(S), \forall S\in \cD, \quad x(S)> v(S) \text{
  otherwise}\}.
\end{align*}
Note that $F(\cD)$ is either empty or a face of the core provided $\cD\ni N$,
and that $F^\circ(\cD)$ is an open face. Define
\[
\mathbf{D} := \{\cD\in\cF\mid \cD\text{ is a sublattice of $\cF$ containing }\emptyset,N,F^\circ(\cD)\neq\emptyset\}.
\]
Observe that any such $\cD$ is necessarily distributive, and therefore is
generated by a poset.
It is easy to see that $\mathbf{D}=\{\cF(x)\mid x\in\core(v)\}$. It follows that
any face of the core is defined by a distributive sublattice of $\cF$. Moreover, the
dimension of a face $F(\cD)$ is $|N|-|h(\cD)|$, where $h(\cD)$ is the height of
the lattice $\cD$.

\subsection{Convex Geometries}
The core of games on convex geometries has been studied by \citet{bileji99}.
\begin{theorem}
Let $v$ be a game on a convex geometry $\cF$.
\begin{enumerate}
\item $\core(v)$ is either empty or a pointed polyhedron
  (i.e., having vertices).
\item Assume that $\core(v)\neq\emptyset$ and that $v$ is nonnegative. Then
  $\core(v)$ is a polytope if and only if $\cF$ is atomistic if and only if $\core(v)=\core^+(v)$.
\end{enumerate} 
\end{theorem}
\begin{remark}
(i) is clear from Theorem~\ref{th:nond} since a convex geometry is
  non-degenerate (since $n$-regular). (ii) was already remarked by
  \citet{fai89} (see Theorem~\ref{th:atom}).  
\end{remark}
A game $v$ is \emph{quasi-convex} if convexity holds only for pairs $A,B\in\cF$
such that $A\cup B\in\cF$. Marginal vectors are defined as usual, considering
all maximal chains in $\cF$ (all of length $n$). 
\begin{theorem}
A game $v$ on $\cF$ is quasi-convex if and only if all marginal vectors belong to the
core. 
\end{theorem}

\subsection{Partition systems}\label{sec:part}
Let $\cF$ be a partition system, $v$ be a game on $\cF$, and $\overline{v}$ its
extension on $2^N$ (see Section~\ref{sec:ext}). If $N\in\cF$, it is
easy to establish that $\core(v)\subseteq\core(\overline{v})$. 

\citet{kawo82} deal with a weaker definition of
partition systems. A partition system only needs to contain all singletons. Then
a \emph{partitioning game} $v$ is a game on $2^N$ defined from some game
$v'$ on $\cF$ by
\[
v(S) = \max\Big\{\sum_{i\in I}v'(F_i), \quad \{F_i\}_{i\in I} \text{ is a
  $\cF$-partition of $S$}\Big\}.
\]
Then $v$ is superadditive and  $\core(v)=\core(v')$ when $N\in\cF$.

\subsection{$k$-regular set systems}\label{sec:kreg}
The core of games on $k$-regular set systems has been studied by \citet{xigr09}.
We mentioned in Section~\ref{sec:gene} that a $n$-regular set system is
non-degenerate, hence Theorem~\ref{th:nond} applies and the core is a pointed
polyhedron, unbounded in general. However, in many cases, $\cF$ could be
degenerate, and in this case the core has no vertices. This is the case for the
2-regular set system given in Figure~\ref{fig:2reg}.
\begin{figure}[htb]
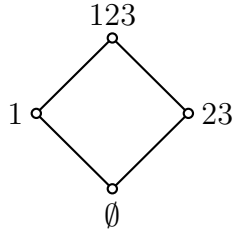

\begin{center}
\psset{unit=0.50cm}
\pspicture(-1,-1)(5,5)
\pspolygon(2,0)(0,2)(2,4)(4,2)
\pscircle[fillstyle=solid](2,0){0.15}
\pscircle[fillstyle=solid](0,2){0.15}
\pscircle[fillstyle=solid](2,4){0.15}
\pscircle[fillstyle=solid](4,2){0.15}
\uput[-90](2,0){$\emptyset$}
\uput[180](0,2){$1$}
\uput[90](2,4){$123$}
\uput[0](4,2){$23$}
\endpspicture
\end{center}
\caption{Example of a degenerate 2-regular set system with 3 players}
\label{fig:2reg}
\end{figure}

Let $\cF$ be a $k$-regular set system, and a
maximal chain $C:=\{\emptyset=S_0,S_1,\ldots,S_k=N\}$. Since $|S_i\setminus
S_{i-1}|>1$ may occur, the classical definition of a marginal worth vector does
not work. Instead, from a given maximal chain, several marginal worth vectors
can be derived. Choose an element $r_i$
in $S_i\setminus S_{i-1}$, $i=1,\ldots,k$. The \emph{marginal worth vector}
associated to $C$ and $r_1,\ldots,r_k$ is defined by
\[
\psi^{C(r_1,\ldots,r_k)} = \sum_{i=1}^k (v(S_i)  - v(S_{i-1}))\1_{r_i}.
\]
We have $\psi^{C(r_1,\ldots,r_k)}(S_i) = v(S_i)$ for all $S_i\in C$.
Denote by $\cM(v,\cF)$ the set of all marginal worth vectors, for all maximal
chains and possible choices of elements. We define the \emph{Weber set} as
\[
\cW(v) := \mathrm{conv}(\cM(v,\cF)).
\]
\begin{theorem}
Let $\cF$ be a $k$-regular lattice closed under union and
intersection. If $v$ is convex,
then $\mathcal{W}(v)\subseteq\core(v)$.
\end{theorem}
\begin{theorem}
Let $\cF$ be a $n$-regular lattice. If $v$ is monotone
  and convex, then $\mathcal{W}(v)\subseteq\core(v)$. 
\end{theorem}
The classical inclusion of the convex part of the core into the Weber set does
not hold in general, as shown by the following counterexample\footnote{This
  example was communicated by J. Derks.}.

Consider the following sets system (regular set lattice but not distributive,
  since it contains a pentagon, figured by the grey circles), with the values  of the game $v$ given into parentheses.
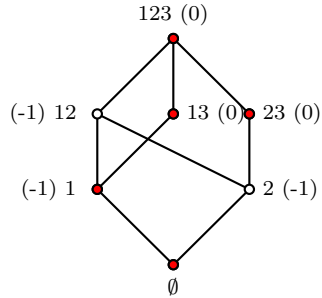
\begin{center}
\psset{unit=0.50cm}
\pspicture(-1,-1)(4,7)
\pspolygon(2,0)(0,2)(0,4)(2,6)(4,4)(4,2)
\psline(0,2)(2,4)(2,6)
\psline(0,4)(4,2)
\pscircle[fillstyle=solid](2,0){0.15}
\pscircle[fillstyle=solid](0,2){0.15}
\pscircle[fillstyle=solid](0,4){0.15}
\pscircle[fillstyle=solid](2,6){0.15}
\pscircle[fillstyle=solid](4,4){0.15}
\pscircle[fillstyle=solid](4,2){0.15}
\pscircle[fillstyle=solid](2,4){0.15}
\pscircle[fillstyle=solid,fillcolor=red](2,0){0.15}
\pscircle[fillstyle=solid,fillcolor=red](0,2){0.15}
\pscircle[fillstyle=solid,fillcolor=red](2,4){0.15}
\pscircle[fillstyle=solid,fillcolor=red](2,6){0.15}
\pscircle[fillstyle=solid,fillcolor=red](4,4){0.15}
\uput[-90](2,0){\scriptsize $\emptyset$}
\uput[180](0,2){\scriptsize (-1) 1 }
\uput[180](0,4){\scriptsize (-1) 12 }
\uput[90](2,6){\scriptsize 123 (0)}
\uput[0](4,4){\scriptsize 23  (0)}
\uput[0](4,2){\scriptsize 2  (-1)}
\uput[0](2,4){\scriptsize 13 (0)}
\endpspicture
\end{center}
The core is defined by
\begin{align*}
x_1 & \geq -1\\
x_2 & \geq -1\\
x_1 +x_2 &\geq -1\\
x_1 + x_3 &\geq 0\\
x_2 + x_3 &\geq 0\\
x_1 + x_2 + x_3 & = 0
\end{align*}
The core is bounded and vertices are $(0,0,0), (-1,0,+1), (0,-1,+1)$.
The Weber set is generated by the marginal vectors associated to the 4 maximal
chains (only 2 differents):
\[
(0,-1,+1), \quad (-1,0,+1).
\]
Clearly, $\cW(v)\not \ni (0,0,0)$.
\begin{remark}\label{rem:web}
As noted in Section~\ref{sec:pocoaug}, the same phenomenon occurs for the
positive core of games on augmenting systems. However, one should be careful
that this kind of (negative) result heavily depends on the definition given to
the marginal vectors: the framework given in Section~\ref{sec:fagr}, which is
more general than augmenting systems, does not exhibit this drawback.    
\end{remark}

\subsection{Coalition structures}
Let $\cB$ be a coalition structure and $v$ on $2^N$ being
zero-normalized (i.e., $v(\{i\})=0$ for all $i\in N$). The core is defined as
follows \citep{audr74}:
\[
\core(v,\cB):=\{x\in\bR^n\mid x(S)\geq v(S),\quad \forall S\in 2^N,\text{ and }
x(B_k)=v(B_k),k=1,\ldots,m\}. 
\]
\begin{theorem}
 Let $\cB$ be a coalition structure and $v$ on $2^N$ being zero-normalized, and
 let $x\in\core(v,\cB)$. Then 
\[
\{y\in\bR^{B_k}\mid (y,x_{|N\setminus B_k})\in\core(v,\cB)\} = \core(B_k,v^*_x,X_k),
\]
with $\core(B_k,v^*_x,X_k) := \{y\in\bR^{B_k}_+\mid y(B_k)=v_x^*(B_k), y(S)\geq
v_x^*(S),\forall S\subseteq B_k,\}$.
\end{theorem}

\subsection{Bounded faces of the core}
We have seen that in many cases, especially if $\cF=\cO(N)$, the core is
unbounded. Since infinite payoffs cannot be used in practice, it is necessary to
find bounded parts of the core and to select a payoff vector in this part. A
simple idea to do this is to impose further restrictions on the definition of
the core. One of them is to impose nonnegativity of the payoff vectors: this
leads to the positive core, studied in the next section. Another one is to
select some inequalities $x(S)\geq v(S), S\in\cF$ of the core and to turn them
into equalities. It is easy to prove that this procedure will eventually lead to
a bounded subset of the core. Equalities $x(S)=v(S)$ can be interpreted as
additional binding constraints: the members of coalition $S$ must share $v(S)$
among them. On the geometrical side, this is nothing else than defining a face
of the core, which has the property to be bounded (see Section~\ref{sec:poly}).

A first attempt in this direction was done by \citet{grxi07} for $\cF=\cO(N)$,
based on an interpretation of $(N,\leq)$ as a hierarchy. Later, the problem was
further elucidated for the case $\cF=\cO(N)$ \citep{gra10a}, and some other
cases (regular set systems and weakly-union closed set systems, although much
less general results can be obtained in these cases). We give here the main results
for the case $\cF=\cO(N)$. 

The basic idea is to suppress all extremal rays of the recession cone
$\core(\cF)$ (recall that the extremal rays do not depend on $v$, but only on
$\cF$), which are of the form $\1_j-\1_i$ for any $i,j\in N$ such that
$j\prec i$ (in other words, each link of the Hasse diagram of $(N,\leq)$
corresponds to an extremal ray and vice versa). We say that an extremal ray
$r\in\core(\cF)$ is \emph{deleted} by equality $x(S)=0$ in $\core(\cF)$ if
$r\not\in \core_{\{S\}}(\cF)$, where
\[
\core_{\{S\}}(\cF) := \{x\in\core(\cF)\mid x(S) = 0\}.
\]
 A first fundamental result is the following.
\begin{proposition}\label{lem:1}
Let $\cF=\cO(N)$, and suppose that $h(N)>0$ (height of $(N,\leq)$). For $i,j\in
N$ such that $j\prec i$, the extremal ray $\1_j-\1_i$ is deleted by equality
$x(S)=0$ if and only if $S\ni j$ and $S\not\ni i$.
\end{proposition}
A collection $\cN\subseteq \cF\setminus\{N\}$ of nonempty sets is said to be
\emph{normal} for $\cF$ if
\[
\core_\cN(\cF):=\{x\in\core(\cF)\mid x(S)=0,\forall S\in \cN\}
\]
reduces to $\{0\}$. We call \emph{restricted core} of $v$ w.r.t. $\cN$ the set
$\core_\cN(v)=\{x\in\core(v)\mid x(S)=v(S),\forall S\in\cN\}$. If nonempty, it
corresponds to a bounded face of $\core(v)$, and the union of all bounded faces
of the core is called the \emph{bounded core}, studied in \citet{grsu12}.

A nonempty normal collection is \emph{nested} if any two sets $S,T$ in the
collection satisfy $S\subseteq T$ or $T\subseteq S$. 

We introduce three examples of normal collections. The first one is introduced in
\citep{gra10a}, and we call it the \emph{upwards collection}. It is obtained as
follows: discard all disconnected elements from $(N,\leq)$, then consider $M_1$
the set of all minimal elements of $(N,\leq)$, and put $N_1:=\downarrow\!M_1$,
the principal ideal generated by $M_1$ (see Section~\ref{sec:poset}). Then on
$N':=N\setminus M_1$, perform the same operations (discard disconnected elements,
consider $M_2$ the set of minimal elements, etc.), till reaching the empty
set. Then $\cN_u:=\{N_1,\ldots,N_q\}$ is a normal collection, minimal in the sense
that no subcollection is normal. The second example comes from \citet{grxi07}
and is built as follows.  First, elements in
$(N,\leq)$ of same height $i$ are put into the \emph{level set}
$Q_{i+1}$. Hence, $N$ is partitioned into level sets $Q_1,\ldots,Q_k$ , and
$Q_1$ contains all minimal elements of $N$. Then 
\[
\cN_{GX}:=\{Q_1,Q_1\cup Q_2,\ldots,Q_1\cup\cdots\cup Q_{k-1}\}
\]
is a nested normal collection. Lastly, a nested collection can be built from
$\cN_u$ (called the \emph{nested closure} of $\cN_u$):
\[
\cN_{ncu} := \{N_1,N_1\cup N_2,\ldots,N_1\cup\cdots\cup N_q\}
\]
(see Figure~\ref{fig:nc} for illustration).
\begin{figure}[htb]
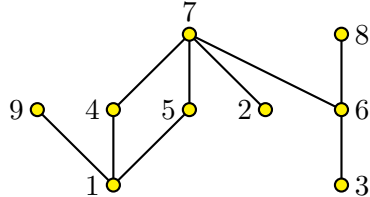

\begin{center}
\psset{unit=0.5cm}
\pspicture(0,-0.5)(8,4.5)
\pspolygon(2,0)(2,2)(4,4)(4,2)
\psline(0,2)(2,0)
\psline(6,2)(4,4)(8,2)
\psline(8,0)(8,4)
\pscircle[fillstyle=solid,fillcolor=yellow](0,2){0.2}
\pscircle[fillstyle=solid,fillcolor=yellow](2,0){0.2}
\pscircle[fillstyle=solid,fillcolor=yellow](2,2){0.2}
\pscircle[fillstyle=solid,fillcolor=yellow](4,2){0.2}
\pscircle[fillstyle=solid,fillcolor=yellow](4,4){0.2}
\pscircle[fillstyle=solid,fillcolor=yellow](6,2){0.2}
\pscircle[fillstyle=solid,fillcolor=yellow](8,0){0.2}
\pscircle[fillstyle=solid,fillcolor=yellow](8,2){0.2}
\pscircle[fillstyle=solid,fillcolor=yellow](8,4){0.2}
\uput[180](0,2){\small 9}
\uput[180](2,0){\small 1}
\uput[180](2,2){\small 4}
\uput[180](4,2){\small 5}
\uput[90](4,4){\small 7}
\uput[180](6,2){\small 2}
\uput[0](8,0){\small 3}
\uput[0](8,2){\small 6}
\uput[0](8,4){\small 8}
\endpspicture
\end{center}
\caption{$(N,\leq)$ has 9 elements. Level 1 is $\{1,2,3\}$, level 2 is
  $\{4,5,6,9\}$ and level 3 is $\{7,8\}$.  The upwards collection is
  $\{123,13456\}$, its nested closure is $\{123,123456\}$, and the Grabisch-Xie
  collection is $\{123, 1234569\}$.}
\label{fig:nc}
\end{figure}

Given a normal collection $\cN=\{S_1,\ldots,S_q\}$, \emph{restricted} maximal
chains are those passing through all sets $S_1,\ldots, S_q$. Marginal vectors
induced by restricted maximal chains are called \emph{restricted marginal
  vectors}, and the \emph{restricted Weber set} w.r.t. $\cN$, denoted by
$\cW_\cN$, is the convex hull of all restricted marginal vectors.
\begin{theorem}
Consider $\cN$ a nested normal collection on $\cF=\cO(N)$. Then for every game
$v$ on $\cF$, $\core_\cN(v)\subseteq \cW_\cN(v)$. In addition, if $v$ is convex,
equality holds.
\end{theorem}
The converse ($\core_\cN(v)=\cW_\cN(v)$ implies convexity of $v$) does not hold
in general. Note that the theorem shows that $\core_\cN(v)$ is nonempty as soon
as $v$ is convex and $\cN$ is nested.

\section{Structure of the positive core}\label{sec:poco}
We first address the question of nonemptiness. Adapting the result of
Bondareva, we say that a collection $\cB$ is \emph{completely balanced} if there
exist positive coefficients $\lambda_B$, $B\in\cB$ such that
\[
\sum_{B\in\cB}\lambda_B\1_B\leq \1_N.
\]
Then a game is \emph{completely balanced} if $\sum_{B\in\cB}\lambda_Bv(B)\leq
v(N)$ holds for every completely balanced collection, and
$\core^+(v)\neq\emptyset$ if and only if $v$ is completely balanced.

An equivalent definition of a completely balanced game is given by \citet{fai89}. 
A game $v$ is completely balanced if and only if  for all families $A_1,\ldots,A_k$ in
$\cF$ and $m\in\bN$ it holds
\[
\frac{1}{m}\sum_{i=1}^k \1_{A_i}\leq\1_N\text{ implies }
\frac{1}{m}\sum_{i=1}^k v(A_i)\leq v(N).
\]
Above, $\cF\ni\emptyset$ is assumed, in order to get the condition $v(N)\geq 0$
by considering the family reduced to $\emptyset$.  
If $\cF$ is closed under intersection and complementation, a nonnegative
balanced game is completely balanced.
\begin{remark}\label{rem:6.1}
The positive core is in general much smaller than the core, and could be empty
even if the core is nonempty. In particular, if $v$ is not monotone, the
positive core is likely to be empty. See also the discussion below on
the equality between the core and the positive core.
\end{remark}

An important question is to know when the core and the positive core coincide. 
\begin{theorem}\label{th:atom}
\citep{fai89} Let $v$ be a nonnegative balanced game on $\cF$ closed under intersection. Then
$\core(v) = \core^+(v)$ if and only if $\cF$ is atomistic. Moreover, $\core(v)$ is
unbounded unless $\core(v)=\core^+(v)$.
\end{theorem}
$\cF$ atomistic implies $\core(v)=\core^+(v)$ is obvious by nonnegativity of
$v$. Also, if $\core(v)\neq\core^+(v)$ then $\cF$ is not atomistic, and for
$\{j\}\not\in\cF$, $x_j$ can be taken arbitrarily negatively large, hence unboundedness.

\subsection{The positive core for augmenting systems}\label{sec:pocoaug}
This has been studied by \citet{bior08a}. 
Given a (nonnegative) game $v$ on $\cF$, we consider its extension $\hat{v}$ on
$2^N$:
\[
\overline{v}(S):=\sum_{T \text{ maximal in } \cF(S)} v(T).
\]
(see Section~\ref{sec:ext}).
Recall that maximal sets in $\cF(S)$ are pairwise disjoint.

Since $N$ does not necessarily belong to $\cF$, the definition of the core is
slightly modified as follows:
\[
\core^+(v):=\{x\in\bR_+^n\mid x(S)\geq v(S) \text{ for all } S\in\cF\text{ and }
x(N)=\overline{v}(N)\}. 
\]
A fundamental (but obvious) property is that $\core^+(v) = \core^+(\overline{v})$ (see
Section~\ref{sec:graph} for a close result, as well as Remark~\ref{rem:cupcap}(v) for a
more general result), and it is a polytope.

Suppose that $N\in\cF$ (hence it is a regular set system). Then any maximal chain
in $\cF$ corresponds to an ordering on $N$ (\emph{compatible}
orderings or permutations). For a maximal chain $C$, denote by $\phi^C$ the
corresponding marginal worth vector. Then the Weber set is naturally defined by
\[
\cW(v) := \mathrm{conv}\{\phi^C\mid C\in\cC(\cF)\}.
\]
Define $v$ to be \emph{convex} if for all $S,T\in\cF$ such that $S\cup T\in\cF$,
\[
v(S\cup T) + \sum_{F\text{ maximal in } \cF(S\cap T)}v(F) \geq v(S) + v(T)
\]
(this is identical to supermodular games on weakly union-closed systems in
Section~\ref{sec:fagr}).
\begin{theorem}
\citep{bior08a} If $v$ is monotone and convex, then $\cW(v)\subseteq \core^+(v)$, and any
marginal vector is a vertex of $\core^+(v)$.
\end{theorem}
The classical inclusion of the core in the Weber set does not hold in general: a
counterexample is given in \citet{bior08a} (see also Remark~\ref{rem:web}).

A game $v$ is superadditive if for all disjoint $S,T\in\cF$ such that $S\cup T\in\cF$,
$v(S\cup T)\geq v(S) + v(T)$.
\begin{theorem}
\citep{bior08a} Let $v$ be a game on $\cF$.
\begin{enumerate}
\item If $v$ is superadditive and monotone, then $\overline{v}$ is superadditive and
  monotone.
\item If $v$ is convex and monotone, then $\overline{v}$ is convex.
\item Suppose $v$ is monotone. Then $v$ is convex if and only if $\overline{v}$ is convex if and only if
  $\core(v)=\cW(\overline{v})$.  
\end{enumerate}
\end{theorem}

\subsection{The positive core and Monge extensions} \label{sec:fagr}
It is possible to get more general results, valid for an arbitrary set system or
a weakly-union closed system, by considering an approach closer to combinatorial
optimization, through the so-called Monge algorithm\footnote{The original idea of the Monge algorithm goes back to \cite{mon1781}. Monge
studied a geometric transportation problem in which a set of locations
$s_1,\ldots,s_n$ of mass points has to be matched optimally (in the sense of
minimizing the total cost) with another set of
locations $t_1,\ldots,t_n$, and proved that optimality was reached if the
transportation lines do not cross. This geometric fact can be expressed as
follows: if the costs $c_{ij}$ of matching objects $s_i$ with $t_j$  have the
``uncrossing'' property:
\[
c_{ij} + c_{k\ell} \geq c_{\max(i,k),\max(j,\ell)} + c_{\min(i,k),\min(j,\ell)}
\]
then the optimal matching is $(s_1,t_1), \ldots,(s_n,t_n)$. This is also called
the ``north-west corner rule''. Translated into the language of set functions, the
uncrossing property is in fact submodularity:
\[
v(A) + v(B) \geq v(A\cup B) + v(A\cap B).
\] 
}. We refer the reader to \citet{fagrhe09} for details and proofs.  

Consider an arbitrary set system $\cF$, and a vector $c\in\bR^n$, which will be
the input vector of the Monge algorithm (MA). The idea of the algorithm is to take at
each iteration the largest subset $F$ of $\cF$ contained in the current set $X$,
and to select in $F$ the first element $p$ corresponding to the smallest
component of a vector $\gamma\in \bR^n$. At initialization, $X=N$ and
$\gamma=c$, and at each iteration, $p$ is discarded from $X$, and $\gamma_p$ is
subtracted from $\gamma_i$, for all $i\in F$. 

The output of the algorithm is the sequence of all selected subsets $F$, the
sequence of all selected elements $p$, and a vector $y\in\bR^{\cF}$ recording at
index $F$ the quantity $\gamma_p$. 

We define $\Gamma(y):=\langle v,y\rangle$.  Letting for any input $c\in\bR^n$
\[
\hat{v}(c) := \Gamma(y),
\]
$\hat{v}$ is an extension of $v$ since it can be proven that $\hat{v}(\1_F) =
v(F)$ for any $F\in \cF$. Moreover,
\[
\core^+(v) = \{x\in\bR^n\mid \langle c,x\rangle\geq \hat{v}(c),\forall c\in\bR^n\}.
\]

The next step is to define marginal vectors, usually defined through
permutations on $N$. The idea here is to take instead the sequence of selected
elements $p$ produced by MA, which is not necessarily a permutation, because
some elements of $N$ may be absent. Let us denote by $\Pi$ the set of all
possible sequences produced by MA, and consider a sequence $\pi\in \Pi$. Then
the marginal vector $x^\pi$ associated to $\pi$ is computed as follows: for each
$p\in \pi$, $x_p^\pi$ is the difference between $v(F)$ (where $F$ is the
smallest selected subset containing $p$) and $\sum_G v(G)$, where the sum is
running over all maximal subsets of $F$ belonging to the sequence. For each
$p\not\in \pi$, we put $x_p^\pi=0$. Clearly, the classical definition is
recovered if $\cF$ is regular, since in this case, the sequence of selected
subsets will form a maximal chain. 

We define the Weber set as
\[
   \cW(v) := \mathrm{conv} \{x^{\pi}\mid \pi\in \Pi\}.
\]
Then it is proved in \citet{fagrhe09} that 
$\core^+(v)\subseteq \cW(v)$. 

\medskip

The last step is to relate equality of the Weber set and the core to
convexity. This is done through the following definition.
A game $v$ on $\cF$ is \emph{convex} if $\hat{v}$ is concave, i.e., it satisfies
for all parameter vectors $c,d\in \bR^N$ and real scalars $0<t<1$,
$$
 t\hat{v}(c) +(1-t) \hat{v}(d)\;\leq\;  \hat{v}(t c + (1-t)d).
$$
\begin{theorem}
Assume $v$ is monotone. Then $v$ is convex if and only if $\core^+(v)=\cW(v)$.
\end{theorem}

The above definition of convexity is done through the extension
$\hat{v}$. However, it is possible to relate it directly to $v$.  
A game $v$ on $\cF$ is \emph{strongly monotone} if for any $F\in \cF$ and
pairwise disjoint feasible sets $G_1, \ldots, G_f\in \cF(F)$ we have
$$
 \sum_{\ell=1}^f v(G_{\ell}) \;\leq\; v(F).
$$
For any intersecting $F, F'\in\cF$ we put
\[
v(F\cap F')  := \sum\{v(G)\mid G\in \cF(F\cap F')\text{ maximal}\}
\]
A game $v$ on $\cF$ is \emph{supermodular} if for all intersecting $F,F'\in\cF$,
we have
\[
v(F\cup F') + v(F\cap F') \geq v(F) + (F').
\] 
\begin{theorem}
A game $v$ is convex if and only if $v$ is strongly monotone and supermodular.
\end{theorem}

\subsection{Extension of $v$ on $2^N$}\label{sec:ext}
Let $\cF$ be weakly union-closed, and $v$ be a game on $\cF$. We introduce an
extension of $v$ on $2^N$ as follows:
\[
\overline{v}(S) = \sum_{T\text{ maximal in } \cF(S)}v(T),\quad \forall S\subseteq N.
\]
\begin{remark}\label{rem:cupcap}
\begin{enumerate}
\item This way of extending a game on $2^N$ appears in
  many different works. For communication graphs, \citet{mye77a} used it
  for computing $v_G(S)$, the extension on $2^N$ of a game $v$ on $G$, for any
  $S\subseteq N$, decomposing $S$ into its (maximal) connected components (in
  this context, see also \citet{owe86}, \citet{boowti92}, \citet{pore95}). It
  can be found also in \citet[\S 5.2]{bil00} with $\cF$ a partition
  system, under the name of \emph{$\cF$-restricted game}, and in
   \citet{bior08a}. In general, it is considered in all the literature on
  communication graphs. This extension has been studied by \citet{fagr09}, and
  arises naturally as the output of the Monge algorithm described in
  Section~\ref{sec:fagr} (see (ii) and (iii) below).
\item Even if $v$ is monotone, $\overline{v}$ need not be monotone. If $v$ is
  monotone, it is not the smallest
extension of $v$ (for this replace $\sum$ by $\max$ in the above equation). If
$\cF$ is union-closed, then $\overline{v}$ is the smallest extension and preserves
monotonicity of $v$ \citep{fagr09}.
\item $\overline{v}$ is given by the Monge algorithm, i.e., $\overline{v}(S)=\hat{v}(\1_S)$ for
  all $S\in 2^N$.
\item The M\"obius transform (see Section~\ref{sec:setsys}) of $\overline{v}$
  vanishes for all $S$ not in $\cF$ (easy fact, remarked by \citet{owe86}). More
  precisely:
\[
m^{\overline{v}}(S) = \begin{cases}
  m^v(S), & \text{for all } S\in\cF\\
  0, & \text{otherwise}
  \end{cases}
\]
where $m^v$ is the M\"obius transform of $v$ on $\cF$. 
\item For any game $v$, we always have
  $\core^+(v)=\core^+(\overline{v})$. Indeed, the inclusion of
  $\core^+(\overline{v})$ in $\core^+(v)$ is obvious. Conversely, assume that
  $x\in \core^+(v)$ and take any $F\not\in \cF$. Then
  $\overline{v}(F)=\sum_{T\text{ maximal in }\cF(F)}v(T)$. We have $x(T)\geq
  v(T)$ for all $T$ maximal in $\cF(F)$. Therefore, since these $T$'s are
  disjoint and $x$ is nonnegative, we find $x(F)\geq
  \overline{v}(F)$. Adapting the previous argument,
  $\core(v)=\core(\overline{v})$ holds provided all singletons belongs to
  $\cF$. Then the maximal sets in   $\cF(F)$ form a partition of $F$ ($\cF$ is a
  partition system).
\end{enumerate}
\end{remark}
\citet{fagr09} have proved the following.
\begin{theorem}
Assume $\cF$ is union-closed, and $v$ is a game on $\cF$. Then $v$ is
supermodular on $\cF$ (in the sense of Section~\ref{sec:fagr}) if and only if
$\overline{v}$ is supermodular on $2^N$.
\end{theorem}

\section{Games on communication graphs}\label{sec:graph}

\subsection{General definitions}
Consider a (undirected) graph $G=(N,E)$, where the vertices are players, and $E$ is
  the set of links. A link between $i,j$ exists if these players can communicate
  or are friends.  Two players are \emph{connected} if there exists a path
  between them.  A \emph{connected coalition} is a subset of $N$ where any two
  players are connected. The set of connected coalitions is denoted by
  $\mathcal{C}_E(N)$.  Maximal connected coalitions of $G$ are called
  \emph{connected components of $G$}, and they partition $N$. The set of
  connected components of $G$ is denoted by $N/E$.  Any coalition $S\subseteq
  N$, even if not connected, can be partitioned into maximal connected
  coalitions (i.e., connected components of the subgraph induced by $S$). The set of
  connected components of $S$ is denoted by $S/E$. This is the framework defined
  by \citet{mye77}.
\begin{remark}
\begin{enumerate}
\item As said in Section~\ref{sec:aug}, set collections induced by communication
graphs are exactly augmenting systems containing all singletons. If the graph is
connected, then they are regular set systems containing all singletons (the
converse is false). Recall also from Section~\ref{sec:weucl} the
characterization of van den Brink, and that these set collections are weakly
union-closed.
\item A generalization of communication graphs is done through conference
  structures of Myerson, or equivalently through hypergraphs (see
  Section~\ref{sec:weucl}).
\end{enumerate}
\end{remark}

A \emph{game on the graph $G=(N,E)$} is a TU-game on $\mathcal{C}_E(N)$ (i.e.,
  it is a game on the collection of feasible coalitions $\cF=\cC_E(N)$). From
  $v$ we define the \emph{extended game} $v_G$ on $2^N$ as follows (see
  Section~\ref{sec:ext}; called \emph{point game} by \citet{boowti92})
\[
v_G(S) = \sum_{T\in S/E}v(T), \forall S\subseteq N.
\]

Since a communication graph may contain several connected components, and
recalling Remark~\ref{rem:augm} (ii), a natural adaptation for the definition of
the core is as follows:
\[
\core(v) := \{x\in\mathbb{R}^n\mid x(C)=v(C),
    \forall C\in N/E, \text{ and }x(S)\geq v(S), \forall
    S\in\mathcal{C}_E(N)\}.
\]
This definition was considered, among others, by
\citet{dem94,dem04}. 
As it is easy to show,  $\core(v)=\core(v_G)$,  which proves that when
nonempty the core is a polytope.
\begin{remark}
Note that if the graph is connected, we recover the definition of the previous
sections, hence all general properties given in Section~\ref{sec:gene} apply.
Concerning the positive core, results in Section~\ref{sec:fagr} apply under the
same condition. Using again Remark~\ref{rem:augm} (ii), the above definition of
the core amounts to take the intersection of all cores on the subsystems
induced by the connected components of $G$.
\end{remark}

We consider below the main families of communication graph, most useful in applications.

\subsection{Communication line-graphs}
Let us assume that the players are ordered according to the natural ordering
$1,\ldots,n$, and consider the set of edges connecting two adjacent players:
$E_0=\{(i,i+1),i=1,\ldots,n-1\}$. Then $G=(N,E)$ is a \emph{line-graph} if $E\subseteq
E_0$, i.e., only some adjacent players can communicate. For convenience, we
introduce the notation $[i,j]:=\{i,i+1,\ldots,j\}$ for $i<j$ in $N$. 

These line-graphs often arise in applications, e.g., water distribution problem
along a river \citep{amsp02}, and auctions situations \citep{grmari90}, and have
been studied by \citet{brlava07}. They show that a sufficient condition for the
nonemptiness of the core is \emph{linear convexity}:
\[
v([i,j]) - v([i+1,j]) - v([i,j-1]) + v([i+1,j-1])\geq 0
\]
for all $[i,j]\in\cC(E)_N$.

\citet{bri09} has characterized communication line-graphs in terms of the
associated set system as follows.
\begin{theorem}
A collection $\cF\subseteq 2^N$ is the set of connected coalitions of a
line-graph if and only if $\cF\ni\emptyset$, $\cF$ is normal (i.e.,
$\bigcup\cF=N$), weakly union-closed, satisfies 2-accessibility (see
Remark~\ref{rem:wuc} (ii)) and \emph{path union stability}.
\end{theorem} 
To explain the last property, we need some definitions. Let $\emptyset\neq
S\in\cF$ and $i\in S$. Then $i$ is an \emph{extreme player} in $S$ if
$S\setminus i\in \cF$. Now, $S$ is a \emph{path} in $\cF$ if it has exactly two
extreme players. The name comes from the fact that a path in $\cF$ corresponds
to a path in the graph (although the converse is false). Path union stability
means that the union of two nondisjoint paths in $\cF$ is still a path in
$\cF$. 

\subsection{Cycle-free communication graphs}
A graph is \emph{cyle-free} if it contains no cycle, in the usual sense of graph
theory.  \citet{browwe92} have characterized this property by what they call
\emph{strong balancedness}: the collection $\cF$ of connected coalitions is
strongly balanced if every balanced collection contains a partition of $N$.

Another characterization is due to \citet{bri09}.
\begin{theorem}
A collection $\cF\subseteq 2^N$ is the set of connected coalitions of a
cycle-free graph if and only if $\cF\ni\emptyset$, $\cF$ is normal (i.e.,
$\bigcup\cF=N$), weakly union-closed, satisfies 2-accessibility (see
Remark~\ref{rem:wuc} (ii)) and \emph{weak path union stability}.
\end{theorem} 
Weak path union stability means that path union stability is required only for
those pairs of paths having a common extreme player.

An important particular case of cycle-free communication graph is the case of
connected graphs. Then the graph is called a \emph{tree}. Games on trees have
been studied by many authors, among them \citet{dem94,dem04}, \citet{helata08},
\citet{khm09}, \citet{babereso08} and \citet{bereso09}. However, most of these
works are more concerned with single-valued solution (as the average tree
solution of \citet{helata08}) than the core (see however \citet{bereso12}).

\subsection{Cycle-complete communication graphs}
A communication graph is \emph{cycle-complete}  if for each cycle of the graph,
the subgraph induced by the players in that cycle is complete (i.e., each player
is connected to every player in the cycle). 

\citet{nobo91} have studied this kind of communication graph. They have shown
that if the game $v$ is convex (assuming $v$ is defined on $2^N$, unlike our
assumption), then $v_G$ is also convex.

\citet{bri09} has characterized cycle-complete communication graphs as follows.
\begin{theorem}
A collection $\cF\subseteq 2^N$ is the set of connected coalitions of a
cycle-free graph if and only if $\cF\ni\emptyset$, $\cF$ is normal (i.e.,
$\bigcup\cF=N$), weakly union-closed, satisfies 2-accessibility (see
Remark~\ref{rem:wuc} (ii)) and the \emph{path property}.
\end{theorem} 
$\cF$ has the path property if for every pair of players $i,j$, there is at most
one path having $i,j$ as extremal players. Alternatively, the path property can
be replaced by closure under intersection.

\section{Acknowledgment}
The author is indebted in particular to Ulrich Faigle, Jean Derks, and Ren\'e van
den Brink for fruitful discussions, and for giving him an incentive to write a
survey on this topic.

\bibliographystyle{plainnat}

\bibliography{../BIB/fuzzy,../BIB/grabisch,../BIB/general}

\end{document}